\documentclass[reqno, 12pt]{amsart}

\usepackage[utf8]{inputenc}
\usepackage{a4wide}
\usepackage{color}
\usepackage{amsmath,amsfonts,latexsym,url}
\usepackage{xcolor}
\usepackage{bm}
\usepackage[colorlinks=true,linkcolor=blue,citecolor=blue,urlcolor=magenta,pdfpagelabels=false,hypertexnames=false]{hyperref}
\usepackage{amsthm}
\usepackage{enumitem}
\usepackage{todonotes}
\usepackage{mathtools}

\usepackage{charter}
\usepackage{microtype}

\usepackage{algorithm}
\usepackage[noend]{algpseudocode}
\algrenewcommand{\algorithmiccomment}[1]{\hfill//\hskip.1em\emph{#1}}
\usepackage[capitalize]{cleveref}

\newtheorem{definition}{Definition}[section]

\newtheorem{lemma}[definition]{Lemma}
\newtheorem{proposition}[definition]{Proposition}
\newtheorem{theorem}[definition]{Theorem}

\newtheorem{remark}[definition]{Remark}
\newtheorem{question}[definition]{Question}

\newcommand{\ord}{{\operatorname{ord}}}
\newcommand{\val}{{\operatorname{val}}}
\def\Q{{\mathbb Q}}
\def\Z{{\mathbb Z}}
\def\Qbar{{\overline{\mathbb Q}}}
\def\Id{{\operatorname{Id}}}

\begin{document}

\title[Minimization of differential equations and algebraic values of $E$-functions]{Minimization of differential equations \\ and algebraic values of $E$-functions}

\author[Alin Bostan]{Alin Bostan}
\address{Alin Bostan, Inria, Universit\'e Paris-Saclay, 1 rue Honor\'e d'Estienne d'Orves, 91120 Palaiseau, France.}
\email{alin.bostan@inria.fr}

\author[Tanguy Rivoal]{Tanguy Rivoal}
\address{Tanguy Rivoal, Institut Fourier, CNRS et Universit\'e Grenoble Alpes, CS 40700, 38058 Grenoble cedex~9.}
\email{tanguy.rivoal@univ-grenoble-alpes.fr}

\author[Bruno Salvy]{Bruno Salvy}
\address{Bruno Salvy, Univ Lyon, EnsL, UCBL, CNRS, Inria,  LIP, F-69342, LYON Cedex 07, France.}
\email{bruno.salvy@inria.fr}

\date{\today}

\begin{abstract}
	A power series being given as the solution of a linear differential 
	equation with appropriate initial conditions, minimization consists in 
	finding a non-trivial linear differential equation of minimal order 
	having this power series as a solution. This problem exists in both 
	homogeneous and inhomogeneous variants; it is distinct from, but related to, 
	the classical problem of factorization of differential operators. 
	Recently, minimization has found applications in Transcendental Number 
	Theory, more specifically in the computation of non-zero algebraic points where 
	Siegel's $E$-functions take algebraic values. We present
	algorithms and implementations for these questions, and discuss examples 
	and experiments.
\end{abstract}



\maketitle

\section{Introduction}
\subsection{Minimization} 
A linear differential equation (LDE)
\begin{equation}\label{eq:defL}
{\mathcal L}(y(z)) \coloneq	a_r(z)y^{(r)}(z)+\dots+a_0(z)y(z)=0
\end{equation}
with polynomial coefficients $a_i(z)$ in $\Q[z]$ is given, together
with initial conditions specifying uniquely a formal power series solution~$S\in\Q[[z]]$, i.e. ${\mathcal L}(S(z))=0$.
In its homogeneous variant, the problem of \emph{minimization} is to find a homogeneous linear differential equation 
\begin{equation}\label{eq:defM}
{\mathcal M}(y(z)) \coloneq b_m(z)y^{(m)}(z)+\dots+b_0(z)y(z)=0
\end{equation}
of minimal possible order~$m$, with polynomial coefficients $b_j(z)$ in~$\Q[z]$ and also having~$S$ as a solution. In the inhomogeneous version, the input is the same, but in the output, a non-zero polynomial right-hand side in~$\Q[z]$ is also possible, which may allow for the existence of an equation of even smaller order. 

Both these problems exist for other fields of coefficients and a large part of our discussion extends to such situations; we focus here on the case of the field~$\Q$  to keep the discussion simple and reflect more closely the capabilities of our implementation.

When the origin is an \emph{ordinary point} of the input equation ${\mathcal L}(y(z))=0$, i.e. $a_r(0) \neq 0$,
initial conditions are given as the vector of values $(S(0),\dots,S^{(r-1)}(0))$ in $\Q^r$.
Otherwise, the origin is called a \emph{singularity} of the equation~\eqref{eq:defL}. 
Equation~\eqref{eq:defL} may still have power series solutions even in the singular case.
The definition of initial conditions in this case is more delicate; it is discussed in \cref{ssec:indicial}.
Interestingly, when the origin is singular, it is even possible that \eqref{eq:defL} possesses a basis of solutions in $\Q[[z]]$; in this case the origin is called an \emph{apparent singularity}. 

\subsection{Algebraic values of \texorpdfstring{$E$}{E}-functions} 
Minimization has a nice application to transcendence theory,
more precisely to the determination of 
algebraic values taken by $E$-functions at algebraic points. 
$E$-functions are entire functions with special arithmetic properties; 
in particular, they satisfy linear differential equations over
$\Qbar(z)$ and their Taylor series at~$0$ belong to $\Qbar[[z]]$.
They have
been introduced  by Siegel in 1929 as a generalization of the exponential
function, and have  been studied in depth by Siegel, 
Shidlovskii, Nesterenko, Andr\'e, Beukers and others. 
We refer to \cref{sec:E-functions} for their
precise definition, statements of
results and bibliographic references. 

Adamczewski and Rivoal
have given an algorithm~\cite{AdamczewskiRivoal2018} that determines the \emph{finite} list of
algebraic numbers $\alpha$ such that $f(\alpha)$ is also algebraic, given as input an
$E$-function $f(z)$ (represented by a linear differential equation and sufficiently many initial terms).
The first two steps of this algorithm rely on the computation of a
minimal homogeneous linear differential equation and of a minimal
linear differential inhomogeneous equation for~$f$.

Another question in number theory concerns the algebraicity
of special values of analytic functions whose Taylor series at~0
belongs to~$\Qbar[[z]]$ and 
satisfies a Mahler equation over
$\Qbar(z)$, i.e., an equation
$\sum_{j=0}^d p_j(z)f(z^{r^j})=0$ where $p_j(z)\in \Qbar[z]$. There
exists an approach based on a similar but different type of
 minimization~%
\cite{AdFaPLMS,AdFa18}, that is not considered here.

\subsection{Relation of minimization to factorization of linear
differential operators}
The problems of factorization and minimization are closely related
since they are both concerned with finding (right) factors of linear
differential operators. But they are different problems. For instance,
an equation can be minimal even when the operator factors. A simple
example is given by the equation $(1-z)y''-y'=0$ and its power series
solution $S=\ln(1-z)$. The corresponding operator clearly has
$\partial_z \coloneq \frac{d}{dz}$ as a right factor of order~1, but
no homogeneous equation of order~1 can have a solution with a
logarithmic singularity. Also, in general, a linear differential
operator may have infinitely many factorizations and the problem of
minimization is to find a minimal (not necessarily irreducible) right
factor that vanishes at the solution~$S$. Thus one cannot simply use
an existing implementation of a factorization algorithm in order to
solve the minimization problem. Even in the case of finitely many
factorizations, minimization can be much simpler than factorization;
this is illustrated by an example in \cref{sec:ex-factor}.

Still, factorization and minimization share many algorithmic tools. Indeed, the algorithm we present in~\cref{sec:minimization} is obtained by combining sub-algorithms of van~Hoeij's description of his factorization algorithms~\cite{Hoeij1997a,Hoeij1997b}, exploiting the fact that 
the situation of minimization is made easier by the extra information provided by the input power series~$S$. 
Furthermore, in applications, it is sometimes possible to take advantage of further structure that the minimal operator is known to possess, such as being Fuchsian (i.e., having only regular singularities, see \S\ref{ssec:singularities}), or having at most one irregular singularity, e.g. at infinity as happens for $E$-functions (see \cref{sec:E-functions}). 

\subsection{Notes on the history of factorization algorithms}
The tools used in factorization and minimization algorithms have a
convoluted history.
Fabry's doctoral thesis~\cite{Fabry1885} is well-known for having completed the classification of the general form that can be taken by formal solutions of linear differential equations at their singularities.
Quite surprisingly, it seems mostly forgotten that the last part of
his thesis (from~Section~32 page~86, to the end,
page~105), Fabry uses the classification result to design a
factorization algorithm for LDEs with rational function coefficients. The principle is to perform a local analysis at each singularity and try all possible choices of local behaviors for the factor under construction. For instance, at any regular singular point, the roots of the indicial equations of the factor have to be roots of the indicial equations of the equations (see \cref{lemma:rightfact}). At irregular singular points, the exponentials must also be chosen from those of the formal solutions. Then, what remains is to find the apparent singularities. In the case of a Fuchsian factor (Sections~32--34), Fabry uses Fuchs' relation (see \S\ref{ssec:Fuchs}) to determine possible roots of the indicial equations at the apparent singularities and then a local  expansion of the logarithmic derivative of the Wronskian of local solutions at a well-chosen singular point reconstructs the required information. This is not completely general as a well-chosen point has to exist; otherwise Fabry introduces undetermined coefficients that will have to be determined later as solutions of a polynomial system. In the non-Fuchsian case (Sections 35--36), he does not have a complete solution in his thesis, but a reduction result showing that  finding a right factor amounts to bounding the number of apparent singularities of that factor.
A few years later~\cite{Fabry1888a,Fabry1888}, he addresses the non-Fuchsian case by observing that the relevant information can be found in the expansion of the logarithmic derivative of the Wronskian at infinity. This anticipates by one century the generalized Fuchs relation of Bertrand and Beukers~\cite[Th. 3]{BertrandBeukers1985}, and its more precise form by Bertrand and Laumon~\cite[A.2]{Bertrand88} described in~\cref{subsec:degree-bounds}.
Again, to make the algorithm complete, it may be the case that undetermined coefficients are needed.
To summarize, Fabry's works contain the first proof that factoring LDEs is algorithmically decidable.

Another approach leading to a general and complete algorithm
 was started by
Markov \cite{Markoff1891b}\footnote{Markov~\cite{Markoff1891b} is
aware of Fabry's note~\cite{Fabry1888}; he writes ``It seems to me that this question can be solved by simpler considerations than those of Mr. Fabry.''} and developed
by Bendixson~\cite{Bendixson1892} and Beke~\cite{Beke1894}\footnote{Most modern references call this approach ``Beke's
algorithm''. Beke~\cite{Beke1894} credits Bendixson~\cite{Bendixson1892} in the
footnote at the end of his article. Both Bendixson and Beke seem
to be unaware of Fabry's and Markov's works.}.
The idea is to construct a $k$th exterior power of the linear differential operator to be factored. This operator must cancel the logarithmic derivative of the Wronskian of the solutions of a factor of order~$k$ (other linear differential equations are constructed for the next coefficients).
One then looks for ``hyperexponential solutions'' (that is, for solutions~$y$ such that~$y'/y$ is rational; these correspond to right factors of order~1). This leads to looking for a rational solution of the associated Riccati equation. As in Fabry's method\footnote{This approach to finding hyperexponential solutions was proposed by Markov~\cite{Markoff1891a,Markoff1891c} and detailed by Bendixson~\cite{Bendixson1892} and Beke~\cite{Beke1894}; a different and more general method is claimed by Painlevé~\cite{Painleve1891}.}%
, this is solved by finding finitely many possibilities at each singularity, gluing them in all possible ways and finding the missing part as a rational solution of  a linear differential equation~\cite[Prop.~4.9]{PuSi03}. 

General bounds on the degree of the coefficients of the factors can be
deduced from both methods; this has been done by Grigoriev for the
Markov--Bendixson--Beke algorithm~\cite[Theorem 1.2]{Grigoriev90} (note
however that the bound is only asymptotic) and by Singer for the
Fabry-type approach~\cite[Lemma 3.7]{Singer93} via the generalized
Fuchs relation in~\cite{BertrandBeukers1985,Bertrand88} (a more
precise version can be found in~\cite[Theorem 1]{BoRiSa21}). The only
algorithmic complexity result we are aware of in the area of
factorization algorithms is due to Grigoriev~\cite[Theorem 1.1]
{Grigoriev90}; it is deduced by studying in depth (improvements of) the Markov--Bendixson--Beke approach.

On the algorithmic front, van Hoeij~\cite{Hoeij1997a} showed that in
many cases, using bounds instead of trying all possible combinations
leads to an improvement to Fabry's approach (which he rediscovered).
This does not give a complete algorithm and the missing cases are
handled by Beke's algorithm. A precise exposition of the computation
of these bounds is the topic of our \Cref{subsec:degree-bounds}. As
for the Markov--Bendixson--Beke approach, various practical
improvements have been proposed in the literature~%
\cite{Schwarz89,Bronstein92,Bronstein94,Tsarev94,Hoeij1997a,ClHo04,JoKaMe13}. Even with these improvements, the Markov--Bendixson--Beke algorithm is not competitive for operators of order larger than~5. In practice, it is outperformed by Fabry-type algorithms such as van Hoeij's~\cite{Hoeij1997a}. ({Another
practical, but incomplete, algorithm is the \emph{eigenring method}, a
Berlekamp-style algorithm introduced by Singer~\cite{Singer96} and
improved by van Hoeij~\cite{Hoeij96}.)

For additional historical information on early contributions on
factoring (by Painlev{\'e}, Fabry, Markov, Bendixson, Beke, etc.) the
reader is invited to consult the following additional references:
\cite[Chap.~III and Chap. IV, \S176, \S177]{Schlesinger1897} for
details of (a variant of) the Markov--Bendixson--Beke algorithm;
\cite[Chap. II, \S10]{Hilb1915} for historical aspects; 
\cite[Letters IX--XII]{Ogigova67} for letters
sent by Hermite to Markov and comments about them; 
\cite[p.~61--123]{Schlesinger1909} for a remarkable bibliography
covering the ``golden age''(1865--1907) of the
theory of linear differential equations.

\subsection{Previous work on minimization}
In contrast to factorization,
it is much more difficult to locate a similar algorithm for
minimization in the literature.

In~\cite[\S6]{Hoeij1997a},
van Hoeij describes an algorithm 
(called ``Construct~$R$'') that solves the following problem:
given a linear differential operator ${\mathcal S}$ with coefficients in $\Q((z))$,
find a non-trivial
right factor of an operator ${\mathcal L}$ with coefficients
in~$\mathbb Q(z)$ known to be a left multiple
of ${\mathcal S}$.
Our minimization algorithm is very close in spirit to that algorithm.
From this perspective, we could say that the first minimization algorithm we are aware of lies ``between the lines'' of~\cite[\S6]{Hoeij1997a}. 

A different, symbolic-numeric, approach to factoring LDEs was proposed
by van der Hoeven~\cite{Hoeven07}. Although minimization is never
explicitly considered in this article, one could consider that it is
implicitly solvable by concatenating several statements from its \S3.3
and \S3.4.
In the same spirit,  the more recent work by Chyzak, Goyer and
Mezzarobba~\cite[\S4]{ChGoMe22} makes this much more explicit,
for Fuchsian input operators~${\mathcal L}$. For instance, a numerical
criterion of
minimality is provided in the Fuchsian case by a numerical computation
of the
monodromy matrices, see lines 1 and~2 of Algorithm~1 
in~\cite{ChGoMe22}.

The proof of \cite[Proposition 8.4]{BoBoKaMe16} 
contains a minimization proof on an explicit and challenging example;
this proof
was the original inspiration of our paper.

Adamczewski and Rivoal proposed in \cite[\S3]{AdamczewskiRivoal2018} a method for minimization based on two ingredients: (i) a priori degree bounds for right factors of ${\mathcal L}$ and (ii) bounds on the order at $z=0$ of linear combinations with coefficients in $\Q[z]$ of $S$ and its derivatives. For (ii), they use multiplicity estimates due to Bertrand and Beukers~\cite{BertrandBeukers1985}, and further refined and made completely explicit by Bertrand, Chirskii and Yebbou~\cite{BeChYe04}. 
For~(i), they use Grigoriev's estimates from \cite[Theorem 1.2]{Grigoriev90}. However, Grigoriev's result is only asymptotic. We gave an explicit and effective bound in~\cite{BoRiSa21}. It is fair to say that the combination of \cite[\S3]{AdamczewskiRivoal2018} and~\cite[Theorem~1]{BoRiSa21} provides the first complete proof that minimization is algorithmically decidable.
However, the corresponding algorithm is highly inefficient in theory, to the point of being completely impractical (see the example in \S\ref{example-Ap31}). 
Our paper can thus be seen as the first one providing an efficient
general algorithm for minimization.

\subsection{Relation to ``order-degree curves''}
To each left
multiple\footnote{Left multiples 
are obtained by multiplication on the left with linear differential
operators
whose 
coefficients belong to~$\mathbb Q(z)$.} with
polynomial coefficients of a linear differential operator~$\mathcal
L$, one associates
the point~$(r',d')$ where $r'$~is the order and~$d'$ the degree. The
(discrete) \emph{order-degree curve} $\mathcal C_\mathcal L$ is
obtained by
keeping those points that lie on the lower part of the convex hull of
this set of points.
In several cases, this
curve has been shown to be well approximated by a hyperbola
for $r$ sufficiently large~%
\cite{BoChLeSaSc07,BoCjCjLi10,ChKa12a,ChKa12b,ChJaKaSi13,Kauers2014}.
This curve has also been studied in
relation to the \emph{desingularisation} ${\mathcal D}$ of ${\mathcal
L}$, which is the lowest-order left multiple of ${\mathcal L}$ with no
apparent singularities~\cite{ChSa98,Tsai00,HuOs00,AbBaHo06,ChKaSi16}. 

Both problems have in common the
remarkable role played by apparent singularities, however
computing ${\mathcal D}$ and computing ${\mathcal M}$ are unrelated
problems. 
The first one is about computing a left \emph{multiple} of ${\mathcal
L}$ with a special property, the second one is about computing a right
\emph{factor}~$\mathcal M$ of ${\mathcal L}$ with a minimality
property.
Although the curve $\mathcal C_\mathcal M$ lies below the
curve~$\mathcal
C_\mathcal L$, bounds on the degree of~$\mathcal M$
are not accessible via the tools developed up to now for~$\mathcal
C_\mathcal L$. Moreover, the degree of~$\mathcal M$ can be
exponentially large with respect to the size of~$\mathcal L$. This is
illustrated by all classical families~$(P_n(z))$ of orthogonal
polynomials, as
they satisfy a small linear differential equation~$\mathcal L(y)=0$ of
order~2 with
coefficients that are polynomial in~$n$ (and thus of bit-size logarithmic
in~$n$), but their
minimal LDE~$P_n(z) y'(z)-P_n'(z)y(z)=0$ has degree~$n$, which is exponential
in the
bit-size of~$\mathcal L$.

\subsection{Contributions}
We give a new minimization algorithm that is efficient in practice.
The necessary tools are presented in detail, with degree bounds
obtained as solutions of explicit linear programming problems, that
we did not find in the literature.
The algorithm for the inhomogeneous case is also new; it reduces the
problem to that of finding rational solutions of the ajoint operator
to the minimal homogeneous one.

Both these minimization algorithms make the
necessary computations for $E$-functions proposed by
Adamczewski and
Rivoal~\cite{AdamczewskiRivoal2018} accessible in practice when the
coefficients are rational.
The final step in their method is based on a process of
desingularization of differential systems,
due to Beukers, that is of independent interest. We make it
more
explicit, with a detailed proof. 
A new canonical decomposition of $E$-functions is presented, together
with corresponding algorithms. 
Explicit families of interesting ``trivial'' evaluations of
hypergeometric ${}_1F_1$ functions at
algebraic points are deduced. 

All these algorithms are practical and an efficient implementation in the computer algebra system
\textsf{Maple} vindicates
them\footnote{Minimization is available as part of the \textsf{gfun} package at \url{https://perso.ens-lyon.fr/bruno.salvy/software/the-gfun-package/}; the code for exceptional algebraic values of $E$-functions can also be downloaded from that page, together with a worksheet of examples.}.

\subsection{Structure of the article} \Cref{sec:minimization} 
describes our
minimization algorithm, both in its homogeneous form (\cref{ssec:hm}) and in its
inhomogeneous variant (\cref{ssec:nhm}). Both crucially rely on computations of degree
bounds, described in \cref{subsec:degree-bounds}.  
\Cref{sec:E-functions} 
discusses the application of the minimization algorithms to
$E$-functions over~$\Q$. The new
algorithm is a practical variant of the Adamczewski-Rivoal algorithm recalled in
\cref{ssec:AR}, itself based on Beukers’ desingularization procedure described and
enhanced in \cref{ssec:beukersalgo}. The algorithm from \cref{ssec:AR} is then applied
in \cref{ssec:decompositionfpqg} to an effective decomposition of 
$E$-functions over~$\Q$. Extensions to $E$-functions with coefficients in a number
field and to $E$-functions in Siegel's original sense are discussed in
\cref{ssec:extensions1,ssec:extensions2}. 
In \cref{sec:examples} we present two infinite families of
${}_1F_1$ hypergeometric functions that take algebraic values at
non-trivial algebraic points (and a similar family of ${}_2F_1$
evaluations, even though they are not $E$-functions).
Finally, \cref{sec:implem} describes our 
implementation of the algorithms and illustrates it with a few timings. 

\section{Minimization algorithm}\label{sec:minimization}

\subsection{Power series solutions}\label{ssec:indicial}
We recall properties of linear differential equations that can be found in the classical treatises of Ince~\cite[Chap.~XVI, XVII]{Ince1956} or Poole~\cite[Chap.~V]{Poole1960}. Moreover, the presentation is specialized to the case of coefficients~$a_i$ of \cref{eq:defL} that are polynomials rather than formal power series.

Given an operator $\mathcal L$ as in \eqref{eq:defL}, 
the image by~$\mathcal L$ of a monomial~$z^s$ with $s\in\mathbb N$ is a polynomial
\begin{equation}\label{eq:inds_L}
f(s,z)=z^{s+g_{\mathcal L}}(p_0(s)+p_1(s)z+\dots+p_t(s)z^t),\qquad -r\le
g_\mathcal L,\quad 0\le t,
\end{equation}
with polynomials~$p_i(s)$ of degree at most~$r$ whose coefficients depend on those of the~$a_i$ and~$p_0\neq0$.
The polynomial~$p_0$ is called  the \emph{indicial polynomial}
of~$\mathcal L$ at~0; we also denote it 
by~$\operatorname{ind}_{\mathcal L}$ and write the identity above as
\begin{equation} \label{indL}
\mathcal L(z^s)\sim \operatorname{ind}_{\mathcal L}(s)z^
{s+g_{\mathcal L}},\quad z\rightarrow 0.
\end{equation}
By linear combination, the image by~$\mathcal L$ of a formal power series $S(z)=\sum_{i\ge 0}c_iz^{i}$ is the formal power series
\[\mathcal L(S)=\sum_{i\ge0}c_if(i,z).\]
The coefficients of~$z^k$ for $k=g,g+1,\dots$ in $\mathcal L(S)=0$ give the equations
\begin{equation}\label{eq:rec1}
c_0p_0(0)=0,\  c_0p_1(0)+c_1p_0(1)=0,\dots,\ 
c_0p_{t-1}(0)+\dots+c_{t-1}p_0(t-1)=0,\end{equation}
{and the linear recurrence of order~$t$}
\begin{equation}\label{eq:rec2}
c_ip_t(i)+\dots+c_{t+i}p_0(t+i)=0,\quad i\ge0.
\end{equation}
These equations imply that the valuation of~$S$ (the index of its first non-zero coefficient) is a zero of the indicial polynomial~$p_0$. Let 
\[\mathcal Z_{\mathcal L}=\{k\in\mathbb N\mid p_0(k)=0\}\]
be the set of nonnegative integer roots of the indicial polynomial of~$\mathcal L$ at~0.
For all $i\not\in\mathcal Z_\mathcal L$, the coefficient~$c_i$ is determined from the previous ones by the $(i+1)$th equation of the infinite system \eqref{eq:rec1}--\eqref{eq:rec2}.
For this reason, the \emph{initial conditions} of the differential equation~\eqref{eq:defL} are the values of~$y^{(i)}(0)$ for~$i\in\mathcal Z_\mathcal L$, as all the other ones are determined by the system \eqref{eq:rec1}--\eqref{eq:rec2}. In the non-singular case when~$a_r(0)\neq0$, the indicial polynomial is $p_0(s)=s(s-1)\dotsm(s-r+1)$ and then $\mathcal Z_\mathcal L=\{0,1,\dots,r-1\}$, recovering the usual definition.
This discussion leads to the following result that will be used to find right factors of~$\mathcal L$.
 (See Prop.~4.3 and Section~4.3 in~\cite{BoLaSa17} for similar considerations.)
\begin{lemma}\label{lemma:rightfact} With the notation above, let $S$ be a power series solution of $\mathcal L$ and 
$\mathcal M$ be a right factor of $\mathcal L$. If there exists a
polynomial $T$ such that $T^{(i)}(0)=S^{(i)}(0)$ for all $i\in
\mathcal Z_\mathcal L$
and~$\mathcal M(T)=O(z^{\max\mathcal Z_\mathcal L+g_\mathcal L+1})$,
then $\mathcal M(S)=0$.
\end{lemma}
\begin{proof} 
We first show that~$\mathcal Z_{\mathcal M}\subset\mathcal Z_
{\mathcal L}$. Indeed, if $\mathcal L=\mathcal A\mathcal M$, then
\[\mathcal L(z^s)=\mathcal A(\mathcal M(z^s))\sim \mathcal A
(\operatorname{ind}_{\mathcal M}(s)z^{s+g_{\mathcal M}})
\sim \operatorname{ind}_{\mathcal M}(s)\operatorname{ind}_{\mathcal
A}(s+g_{\mathcal M})z^{s+g_{\mathcal M}+g_{\mathcal A}},\quad z\rightarrow 0,\]
which, in combination with~\eqref{indL}, implies that
\[
\operatorname{ind}_{\mathcal L}(s) = \operatorname{ind}_{\mathcal M}(s) \operatorname{ind}_{\mathcal A}(s+g_{\mathcal M}) \quad \text{and} \quad 
g_{\mathcal L} = g_{\mathcal M}+g_{\mathcal A}.
\]
In particular, the indicial polynomial of~$\mathcal M$ divides that
of~$\mathcal L$, and hence $\mathcal Z_{\mathcal M}\subset\mathcal Z_
{\mathcal L}$.

Applying the discussion above to~$\mathcal M$ shows that
the coefficients of~$T$ satisfy the first $\max\mathcal Z_\mathcal
L+1$ equations of the system \eqref{eq:rec1}--\eqref{eq:rec2}. 
It follows that
$T$ can be extended to a unique power series solution of~$\mathcal M$. As~$\mathcal M$ is a right factor of~$\mathcal L$, this power series is also a solution of~$\mathcal L$. Since it has the same initial conditions as~$S$, they coincide.
\end{proof}

\subsection{Homogeneous minimization} \label{ssec:hm}

Since initial conditions are given for the power series~$S$ solution of the linear differential equation, it is possible to compute arbitrarily many coefficients of~$S$. The algorithm relies on the computation of upper bounds on the degree of the coefficients of right factors  of the linear differential operator of a given order. Given such bounds and sufficiently many coefficients of~$S$, it is easy to set up a (structured) linear system whose solutions are the possible coefficients of a right factor, or only~0 if no such factor exists. When a non-zero solution is found, one takes its greatest common right divisor with the original linear differential operator and checks it using~\cref{lemma:rightfact}. This approach is described in \cref{algo:min-hom}.

\begin{algorithm}
\caption[]{Minimal right factor}\label{algo:min-hom}
\begin{algorithmic}[1]
\Require{$\mathcal L=a_r(z)\partial_z^r+\dots+a_0(z)$ in $\Q[z]\langle\partial_z\rangle$;

ini: $S_0$ a truncated power series at precision $\ge \max\mathcal Z_\mathcal L$

\qquad specifying a unique solution~$S\in\Q[[z]]$ of $\mathcal L(S)=0$.}
\Ensure{a right factor of~$\mathcal L$ in~$\Q[z]\langle\partial_z\rangle$ of minimal order that vanishes at~$S$}
\State $\mathcal M \coloneq \mathcal L$;\,
 $T\coloneq S_0$;\,
 $m\coloneq r$;\, $p\coloneq \max\mathcal Z_\mathcal L+r;$\label{algo-line1}
    \While{$m>1$}
		\State      $m\coloneq m-1$
    	\If{$N\coloneq $\Call{BoundDegreeCoeffs}{$\mathcal L$, $m$}$\neq${FAIL}}
			\While{true}
				\State $T\coloneq $\Call{SeriesSolution}{$\mathcal
				L,T,p+m$};\,$k\coloneq \lfloor p/(m+1)\rfloor;$
					\label{algo-line7}
				\State $\mathcal H\coloneq \Call{ApproximantBasis}
					{T,T',\dots,T^{(m)};k,\dots,k;p};$
					\label{algo-line8}
        	   	\If{$\mathcal H=\emptyset$ and $p\ge (m+1)(N+1)$}{\label{algo-line9}
		           break}\Comment{No right factor of order~$m$}\EndIf
				\If{$\mathcal H\neq\emptyset$}\Comment{$\mathcal H$
					contains at least a candidate factor $h$}
					\State $\mathcal G\coloneq \Call
					{GreatestCommonRightDivisor}{\mathcal L,h}$;
	           		\If{$\mathcal G(T)=O(z^{\max\mathcal Z_\mathcal
           					L+g_\mathcal L+1})$\label{algo-line12}}
           					{ $\mathcal
           					M\coloneq \mathcal G;\,m\coloneq \ord\,\mathcal M$
						\label{algo1-line14};\,break}
				\EndIf\EndIf
           		\State $p\coloneq 2p$\label{algo-line15}
			\EndWhile
        \EndIf
    \EndWhile
\State\Return{$\mathcal M$}
\end{algorithmic}
\end{algorithm}

It relies on several other algorithms that we now review.

\subsubsection{Sub-algorithms}

\paragraph{\sc SeriesSolution} Takes as input a linear differential operator, a truncated power series 
solution of it, and a target precision~$p$. It returns the power series solution of the operator up to~$O(z^p)$, obtained either by truncating the power series given as input, or by extending it using the linear recurrence deduced from the differential equation.

\smallskip
\paragraph{\sc ApproximantBasis} Takes as input ${k}$ power series~$(S_1,\dots,S_k)$ that are the truncations at precision~$p$ of the successive derivatives of~$S$; $k$ nonnegative integers~$(s_1,\dots,s_k)$ and the precision~$p$. It first computes a basis $B(z)\in\mathbb Q[z]^{k\times k}$ of the $\mathbb Q[z]$-module
\begin{equation}\label{eq:PHmodule}
\mathcal A_p\coloneq \{(p_1,\dots,p_k)\mid p_1S_1+\dots+p_kS_k=O(z^p)\}
\end{equation}
in \emph{shifted Popov form}~\cite{Popov1972,BarelBultheel1992,BeckermannLabahnVillard1999,JeannerodNeigerVillard2020} with shift vector~$(-s_1,\dots,-s_k)$. This
implies that 
any element~$P$ of~$\mathcal A_p$ with degrees bounded by~$(s_1,\dots,s_k)$ is a linear combination of the rows of~$B$ whose index~$i$ satisfies $\deg B_{ii}\le s_i$. Those are the rows returned by~{\sc ApproximantBasis}. As the~$S_i$ are successive
derivatives, these rows can be interpreted as linear differential
operators~$p_1+p_2\partial+\dots+p_k\partial^{k-1}$.
Efficient algorithms to compute such bases are 
known~\cite{JeannerodNeigerVillard2020}. 

\smallskip
\paragraph{\sc GreatestCommonRightDivisor} Computes the monic greatest
common right divisor (gcrd) of two linear differential operators with
coefficients in~$\mathbb Q(z)$. This is classically
achieved by a non-commutative version of Euclid's algorithm~\cite{Ore1933} and
more efficient methods are known~\cite{Grigoriev90,Hoeven16}. 

\smallskip
\paragraph{\sc BoundDegreeCoeffs} This is the heart of the algorithmic work, described in \cref{subsec:degree-bounds}. It takes as input an operator of order~$r$ and a positive integer~$m<r$. It returns either~FAIL when it has proved that no right factor of order~$m$ with polynomial coefficients exist; or an upper bound on the degree of each of the coefficients such a factor would have.

\begin{theorem} Given a linear differential operator~$\mathcal L\in\mathbb Q[z]\langle\partial_z\rangle$ and a truncated power series specifying a unique solution~$S\in\mathbb Q[[z]]$ of~$\mathcal L(S)=0$, 
\cref{algo:min-hom} computes a non-zero right factor~$\mathcal M$ of~$\mathcal L$ of minimal order such that~$\mathcal M(S)=0$. \end{theorem}
\begin{proof}
\quad 1. (Correctness assuming termination.)
Since $T$ is expanded at precision~$p+m$ in \cref{algo-line7} and
$p>\max \mathcal Z_\mathcal L$ from \cref{algo-line1,algo-line15}, it satisfies $T^{(i)}(0)=S^{(i)}(0)$ for $i\in\mathcal Z_\mathcal L$. 
In \cref{algo-line8}, all series $T,T',\dots,T^{(m)}$ are known at
precision~$p$. It follows that if the basis returned by {\sc
ApproximantBasis} is empty with the given bounds on the degrees of the
coefficients in \cref{algo-line9}, there is no right-factor of~$\mathcal L$ of
order~$m$. Otherwise, taking~$\mathcal G$ a gcrd of~$\mathcal L$ and
an element~$h$ of~$\mathcal H$ gives a right factor of~$\mathcal L$ to
which \cref{lemma:rightfact} applies, showing that $\mathcal M(S)=0$
if the condition on \cref{algo-line12} holds. The loop on~$m$ makes
the algorithm
terminate on a right factor of minimal order.

\qquad 2. (Termination.)
The only possible source on non-termination in the algorithm is the loop where $p$ is doubled every time $\mathcal G$ fails to cancel~$T$ to sufficient precision. Let $V_p$ be the $\mathbb Q$-vector space generated by the approximants of the modules~$\mathcal A_{p'}$ from \cref{eq:PHmodule} for all~$p'\ge p$. Since the approximants have degrees bounded by~$(N,\dots,N)$, these are finite-dimensional vector spaces and~$V_{p+1}\subset V_p$. Thus there exists~$p_0$ such that~$V_{p_0}$ is the intersection of all~$V_p$ for~$p\ge p_0$. Any approximant~$h=(h_0,\dots,h_{m})$ in~$\mathcal H$ in \cref{algo-line8} for~$p\ge p_0$ has the property that $h_0S+\dots+h_mS^{(m)}=O(z^k)$ for all~$k\ge p$ and thus annihilates~$S$, and therefore so does its gcrd with~$\mathcal L$, making the algorithm terminate.
\end{proof}

\subsubsection{Comparison with van Hoeij's algorithm}
Van Hoeij's Algorithm ``Construct $R$''~\cite[p.~552]{Hoeij1997a} follows a similar pattern. Our termination proof is essentially his. The difference is that instead of looking for an arbitrary right factor of~$\mathcal L$, we need to make sure that the factor returned by the algorithm cancels the power series~$S$. This is ensured by the test in \cref{algo-line12}.

\subsubsection{Example}\label{example-Ap31}
Consider the sequence
\[u_n=\sum_{k=0}^n{\frac{n!(n+k)!}{k!^4(n-k)!^3}}.\]
Zeilberger's creative telescoping algorithm~\cite{Zeilberger91} shows that~$u_n$ satisfies a linear recurrence of order~4 with coefficients that are polynomials in~$n$ of degree at most~10:
\begin{footnotesize}
\begin{multline*}
\left(29412 n^{4}+224352 n^{3}+632931 n^{2}+781692 n +356309\right) \left(n +3\right)^{2} \left(n +4\right)^{4} u_{n +4}\\
+\dots+4 \left(29412 n^{4}+342000 n^{3}+1482459 n^{2}+2838258 n +2024696\right) \left(n+1\right)^{2} u_n=0
.\end{multline*}
\end{footnotesize}
This recurrence translates into a linear differential operator of order~$10$ annihilating the generating function~$S(z)=\sum_{n\ge0}{u_nz^n}$, with coefficients of degree at most~8:
\begin{footnotesize}
\begin{multline}\label{eq:deq-Apery31}
\mathcal L=29412 z^{8} \partial_z^{10}
-684 z^{7} \left(688 z -1489\right)\partial_z^{9}
-21 z^{6} \left(156864 z^{2}+742368 z -588707\right) \partial_z^{8}+\dotsb\\
+\left(99370416 z^{3}-1926228512 z^{2}-19342508 z +8500\right)\partial_z
+4 \left(2024696 z^{2}-3141504 z -32725\right).
\end{multline}
\end{footnotesize}
The only integer roots of the indicial polynomial of~$\mathcal L$ at~0 are~in~$\mathcal Z_\mathcal L=\{0,1\}$ so that the initial conditions specifying~$S$ uniquely are~$S(0)=u_0=1, S'(0)=u_1=3$. The differential operator~$\mathcal L$ is not minimal for~$S$.
There are two stages in the execution of the algorithm: first, a right
factor is sought; next, its minimality is proved. 

In the first stage, tight bounds on the degrees of coefficients of
right factors are not needed. One can compute more and more
coefficients of the series expansion of the solution and try to
reconstruct a factor by computing an approximant basis. When a
non-trivial factor exists, it will be discovered.

In the second stage, or if no non-trivial factor exists, i.e., if
$\mathcal L$ is minimal, then one has to certify this minimality. This
is where tight bounds are useful. In this example, the bound on the degree of the coefficients that
follows from the work of Bertrand, Chirskii and 
Yebbou~\cite[Lemma~3.1]{BeChYe04}
is larger than
\[10^{10^{10^{33}}}.\]
This makes it a purely theoretical result that cannot be used in a
computation. Indeed, with current implementations and hardware,
already bounds
on degrees of order~$10^7$ become too large for practical
computations.

\medskip
\paragraph{\em Computation of a right factor.}
Not knowing in advance that $\mathcal L$ is not minimal, our algorithm
first computes bounds on the degrees of the coefficients of right
factors. During this computation of bounds, it discovers that
$\mathcal L$ does not have any right factor of order~9, 8, or~7.
For order~$6$, a bound~30 for the degrees of the coefficients is
found. With this bound, a candidate linear differential operator of
order~6 and degree~8 is found:
\begin{footnotesize}
\begin{multline*}
\mathcal M=z^{4} \left(1882368 z^{4}-2206584 z^{3}+1703460 z^{2}+67815 z +272\right)\partial_z^6+\dotsb+\\
+2 \left(3764736 z^{6}-41001696 z^{5}+157022376 z^{4}-184937064 z^{3}-6917519 z^{2}-3408891 z -41888\right).
\end{multline*}
\end{footnotesize}
The computation of the greatest common right factor stops at its first step, discovering that this differential operator is a right factor of~$\mathcal L$.

\medskip
\paragraph{\em Proof of minimality.}
This is the stage where good degree bounds are useful. At order~5, the
bound given by our previous
work~\cite[Thm.~1]{BoRiSa21}, using the
generalized exponents and the slopes of the Newton polygons
of~$\mathcal L$, is only~87, to be compared
with the purely theoretical bound above.
This means that proving
minimality reduces to the computation of~$88\times 6+1=529$
coefficients of the power series followed by a computation of an
approximant basis. This is a quantity that is manageable, but
motivates the quest for tight degree bounds.

At this same order~5, our algorithm computes the better bound~$N=15$ 
(instead of~87). Thus, with $p+1$ coefficients of~$S$, where
$p=96=6\times 16$, the computation of {\sc ApproximantBasis} shows
that there is no non-zero operator~$h$ of order~5 with coefficients of degree at most~$N$ such that~$h(S(z))=O(z^p)$ and therefore no right factor of~$\mathcal L$ of order~5 annihilating~$S$. 

Next,  the bound on the degrees of the coefficients of a right factor of order~$4$ is smaller than~$15$, so that if a right factor of that order existed, it would have been obtained for order~5. Finally, the computations of bounds for orders~$3,2,1$ show that no factor of these orders exist. This concludes the proof of minimality of the operator~$\mathcal M$ for~$S$.

\subsection{Degree bounds}\label{subsec:degree-bounds}
The computation of degree bounds for a factor of a given order is a
key step in van Hoeij's factorization 
algorithm~\cite[\S9]{Hoeij1997a}. We recall the ingredients here. Compared to our earlier
work~\cite{BoRiSa21} where we have obtained universal bounds, the
bounds computed here are tailored to the equation under study, rather
than depending only on its order, degree and height. This allows for
smaller bounds and more efficient computations. As shown in the
example above, having good bounds is important when certifying the
minimality of a right factor. In this work, this is achieved by
setting up explicit integer linear programming problems that do not
appear in the earlier literature.

\subsubsection{Singularities of the factors}\label{ssec:singularities}

Dividing~$\mathcal L$ by its leading coefficient~$a_r$ gives a monic operator with rational function coefficients. In this form, the singularities of~$\mathcal L$ are the poles of its coefficients. A singularity~$\alpha$ of~$\mathcal L$ is called \emph{regular} 
if the indicial polynomial of~$\mathcal L$ at~$\alpha$ has degree equal to the order~$r$ of~$\mathcal L$, and it is called \emph{irregular} otherwise. 
The right factors will be searched in the same monic form. Recall that the valuation~$\val_\alpha(r)$ of a rational function~$r$ at~$\alpha$ is the exponent of the leading term of the Laurent expansion of~$r$ at~$\alpha$ (and $\val_\alpha(0)=\infty)$. At a regular singularity the valuation of each coefficient~$a_i$ of~$\mathcal L$ in monic form is at least~$i-r$. 
Bounds on the degrees of the coefficients of factors are obtained by bounding the valuations of their coefficients in monic form at each singularity and at infinity, and by bounding the number of \emph{apparent} singularities. Apparent singularities are poles of the coefficients where the operator has a basis of $r$~formal power series solutions; they are regular. All these notions are classical and can be found for instance in Ince's book~\cite{Ince1956}.

\subsubsection{Newton polygons and valuations of the coefficients of the factors}

The Newton polygon of the operator~$\mathcal L$ from \cref{eq:defL} at~0 is the convex hull of the union of the quadrants~$(i,\val_0(a_i)-i)+(\mathbb{R}_{\le0}\times\mathbb{R}_{\ge0})$. The knowledge of the Newton polygon of~$\mathcal L$ at~0 gives lower bounds on the valuations of its coefficients.
The main property of relevance here is that the Newton polygon of a product of operators is the (Minkowski) sum of their Newton polygons 
(\cite[Lemme~1.4.1]{Malgrange1979}). For instance, when~0 is an ordinary point or a regular singularity of~$\mathcal L$, the only slope of the lower part of its Newton polygon is~0 and this is therefore a property of the Newton polygons of the monic factors of~$\mathcal L$, which reflects the fact that they are regular at~0 in that case.

More generally, let $(x_0,y_0)=(0,y_0),\dots,(x_k,y_k)=(r,y_k)$ be the points on the lower part of the Newton polygon of~$\mathcal L$ and let~$((n_1,d_1),\dots,(n_k,d_k))$ with $(n_i,d_i)=(x_i-x_{i-1},y_i-y_{i-1})$ be the tuple
of segments of the Newton polygon of~$\mathcal L$ sorted by increasing slope. Then the lowest possible Newton polygon for a monic factor of order~$m$ is obtained 
from the solution of the ``0-1 knapsack problem''
\[\min\sum_{i=1}^k{c_in_i}\quad\text{subject to }\sum_{i=1}^k{c_id_i}=m\text{ and } c_i\in\{0,1\},\quad i=1,\dots,k,\]
where~$c_i$ is either 1 or 0 depending on whether or not the slope~$(n_i,d_i)$ is used. This solution allows one to obtain lower bounds on the valuations at~0 of the coefficients of monic factors of~$\mathcal L$ of order~$m$. The 0-1 knapsack problem is NP-hard but lower bounds can be found efficiently if needed~\cite[Ch.~8]{Vazirani2001}. In practice, this has never been a costly step in our computations and an optimal value may lead to a better degree bound which saves computation time in other steps of the minimization algorithm.

The same process can be performed at every irregular singularity~$\alpha$ of~$\mathcal L$ by considering the Newton polygon formed from~$\val_\alpha$ instead of~$\val_0$. Thus lower bounds on the valuations of the coefficients of a factor are found at each singularity, from the Newton polygon of~$\mathcal L$ and the order of the factor. Applying the same process at~$\infty$ (for instance by changing $z$ into~$1/z$ and working at~0) gives bounds on the valuation at infinity of these coefficients.

\subsubsection{Fuchs' relation and apparent singularities of the factors}
\label{ssec:Fuchs}
Let $\mathcal M$ denote a monic right-factor of order~$m$ of the operator~$\mathcal L$ to be minimized.
The study of the Newton polygons of $\mathcal L$ provides lower bounds on the valuations of the coefficients of~$\mathcal M$ at the singularities of~$\mathcal L$.
We now show how to obtain an upper bound on the number of apparent singularities of~$\mathcal M$; together with the lower bounds on valuations, this will provide upper bounds on the degrees of (the polynomial version of)~$\mathcal M$.

We first recall the principle of the method in the case where 
$\mathcal M$ is Fuchsian, that is, if all its singularities (including $\infty$) are regular.
Fuchs' relation~\cite[p.~138]{PuSi03} states that 
\begin{equation}\label{eq:Fuchs}
\sum_{\rho\in\operatorname{Sing}(\mathcal M)}S_\rho(\mathcal M)=-m(m-1),
\end{equation}
where $\operatorname{Sing}(\mathcal M)$ is the set of singularities of~$\mathcal M$, including the apparent ones and infinity, and where
\begin{equation}\label{eq:srho}
S_\rho(\mathcal M)\coloneq \sum_{j=1}^m{e_j(\rho)}-\frac{m(m-1)}2,
\end{equation}
the numbers~$e_j(\rho)$ being the local exponents of~$\mathcal M$ at the point~$\rho$ (they are the roots of the indicial polynomial at~$\rho$). At an apparent singularity~$\rho$, the quantity~$S_\rho(\mathcal M)$ is a positive integer, so that the number of apparent singularities is upper bounded by
\begin{equation}\label{eq:bound-app-Fuchs}
\#\!\operatorname{App}(\mathcal M)\le -m(m-1)-\sum_{\rho\in\sigma(\mathcal M)}S_\rho(\mathcal M),
\end{equation}
with~$\sigma(\mathcal M)$ the subset of~$\operatorname{Sing}(\mathcal M)$ formed by the singularities of $\mathcal M$ that are not apparent.

Since~$\mathcal M$ is a right factor of~$\mathcal L$, the set $\sigma
(\mathcal M)$ is a subset of~$\operatorname{\sigma(\mathcal L)}$. The
set~$\operatorname{\sigma(\mathcal L)}$ is known, since it corresponds
to the roots of the leading coefficient~$a_r(z)$ that are not apparent
singularities of~$\mathcal L$, plus possibly $\infty$. Let $\mu_
{1},\dots,\mu_{s}$ be the irreducible factors of~$a_r(z)$
corresponding to non-apparent singularities of~$\mathcal L$ and by
convention let~$\mu_0=z$. At a finite $\rho\in\sigma(\mathcal L)$,
given by its minimal polynomial~$\mu_i$, the indicial polynomial $
\operatorname{ind}_\rho^{\mathcal L}(\theta)\in\Q(\rho)[\theta]$ is
easily computed. Then the unknown indicial polynomial~$
\operatorname{ind}_\rho^{\mathcal M}(\theta)\in\Q(\rho)[\theta]$ has
to be a factor of~$\operatorname{ind}_\rho^{\mathcal L}(\theta)$ of
degree exactly~$m$. Let $I_{i,j}(\theta)$, $j=1,\dots,k_i\le r$ be the
irreducible factors of $\operatorname{ind}_\rho^{\mathcal L}(\theta)$
in $\Q(\rho)[\theta]$, repeated with their multiplicity (and
similarly~$I_{0,j}(\theta)$ denote the factors of the indicial
polynomial of~$\mathcal L$ at infinity). The sum of the roots of $I_
{i,j}$ lies in~$\Q(\rho)$ and its sum over all roots of~$\mu_i$ is a
rational number~$e_{i,j}$. A bound on the number of apparent
singularities is therefore obtained by solving the following integer
linear programming problem~\cite[Part~VI]{Schrijver1986}
\begin{gather*}
\max A\quad\text{subject to }\ A=-m(m-1)-\sum_{i=0}^s\deg\mu_i\left(\sum_{j=1}^{k_i}{c_{i,j}e_{i,j}}-\frac{m(m-1)}2\right)\in\mathbb{N}\\
\text{and for all $i\in\{0,\dots,s\}$,}\quad\sum_{j=1}^{k_i}{c_{i,j}\deg I_{i,j}}=m,\quad c_{i,j}\in\{0,1\}.
\end{gather*}
The constraints express the fact that there should be  $m$ exponents at each root of~$a_r$, be them singular or ordinary for~$\mathcal M$.

This process can be used whenever the right factor~$\mathcal M$ to be found is known to be Fuchsian, thus in particular when~$\mathcal L$ itself is Fuchsian.

Note that if there is no choice of~$c_{ij}$ for which~$A\in\mathbb{N}$, then there is no right factor of order~$m$. 
As for the previous optimization problem, this is potentially a
computationally expensive step. Simple upper bounds can be obtained by
solving the relaxed linear programming problem where the constraints
$0\le
c_{i,j}\le 1$ replace the binary variables.

\subsubsection{Generalized Fuchs relation}
To an irregular singular point $\rho$ of~$\mathcal L$ is associated a set of \emph{exponential parts}. 
If $\rho$ is finite, these are polynomials~$w(z)$ in some rational power~$1/r$ of~$z$ ($r\in\mathbb{N}_{>0}$) such that $\mathcal L$ admits a formal solution of the form 
\[\exp\!\left(\int{\frac{w(1/(z-\rho))}{z-\rho}\,dz}\right)S(z),\qquad S\in\Q\!\left[\left[(z-\rho)^{1/r}\right]\right][\log(z-\rho)],\quad\operatorname{val}_{z=\rho}S=0.\]
The case when $\rho=\infty$ is obtained by changing $z$ into~$1/z$ in
the equation. That a full basis of formal solutions can be obtained in
this way goes back to Fabry's classification~\cite{Fabry1885}.
Algorithms for the computation of the list of exponential parts at a
point~$\rho$ have been introduced early in computer algebra~\cite%
[Thm.~4.2.1 and Cor.~4.3.1]{Malgrange1979}, \cite{DeDiTo92,Hoeij1997b}.
When $\rho$ is a regular singular point, the exponential parts are constants that coincide with the roots of the indicial polynomial. In general, the \emph{generalized local exponents} are the constant coefficients of the exponential parts. If each of them is counted with multiplicity~$r$, their number is exactly the order of the operator.

When $\mathcal M$ is a right factor of~$\mathcal L$, its exponential parts at~$\rho$ form a subset of those of~$\mathcal L$. 
If the order of~$\mathcal M$ is~$m$, the Fuchs relation~\eqref{eq:Fuchs} generalizes as 
\begin{equation}\label{eq:generalized_Fuchs}
\sum_{\rho\in\operatorname{Sing}(\mathcal M)}\!\left(S_\rho(\mathcal M)-
\frac12 I_\rho(\mathcal M)\right)
=-m(m-1),
\end{equation}
where~$S_\rho$ is as in \cref{eq:srho}, with the generalized local exponents taking the place of the local exponents and 
\[
I_\rho(\mathcal M)
\coloneq  2\sum_{1\le i<j\le m}\deg(w_i-w_j),
\]
where~the $w_i$ are the exponential parts at~$\rho$,
see~\cite[Thm.~2 and \S5]{Bertrand1999}, \cite[p.~84]{Bertrand88}. As
the $w_i$
are polynomials
in a
fractional power of~$z$, their degree here is a rational number. (The quantity $I_\rho(\mathcal M)$ is related to Malgrange's irregularity of $\mathcal{M}$ at $\rho$; see \cite[\S2.2]{BoRiSa21} for details which are not essential here.)

Thus, the analogue of \cref{eq:bound-app-Fuchs} is
\begin{equation}\label{eq:nb-app-irr}
\#\!\operatorname{App}(\mathcal M)\le -m(m-1)-\sum_{\rho\in\sigma(\mathcal M)}\left(S_\rho(\mathcal M)-
\frac12 I_{\rho}(\mathcal M)\right).
\end{equation}
The corresponding optimization problem is slightly more involved. As in the Fuchsian situation, $\sigma(\mathcal M)\subset\sigma(\mathcal L)$ and we denote by~$\mu_1,\dots,\mu_s$ the irreducible factors of the leading coefficient~$a_r(z)$ that correspond to non-apparent singularities of~$\mathcal L$ and $\mu_0=z$. At a finite~$\rho\in\sigma(\mathcal L)$, given by its minimal polynomial~$\mu_i$, the exponential parts are given as
\[w_{i1}((z-\rho)^{1/r_{i1}}),\dots,w_{ik_i}((z-\rho)^{1/r_{ik_i}})\]
with minimal~$r_{ij}$. Each contributes $r_{ij}$ times its constant coefficient to the set of generalized local exponents of~$\mathcal L$ at~$\rho$, so that $\sum{r_{ij}}=\operatorname{\ord}(\mathcal L).$
The exponential parts of~$\mathcal M$ at~$\rho$ form a subset of those
of~$\mathcal L$. This property, combined with the generalized Fuchs
relation~\eqref{eq:generalized_Fuchs}, leads to the following
integer linear programming problem
\begin{gather*}
\max A\quad{\text{subject to}}\\
\begin{split}
A=-m(m-1)
-\sum_{i=0}^s\sum_{\mu_i(\rho)=0}\left(\sum_{j=1}^
{k_i}{{c_{i,j}(r_{i,j}w_{i,j}(0)-\frac{r_{ij}(r_{ij}-1)}2\deg w_{ij})-\frac{m(m-1)}2}}\right)\\
+\sum_{i=0}^s\deg\mu_i\sum_{j=1}^{k_i}\sum_{1\le k\neq j\le k_i}d_{i,\{j,k\}}\deg(w_{ij}-w_{ik})
\in\mathbb N
\end{split}\\
\text{and for all $i\in\{0,\dots,s\}$,}\quad\sum_{j=1}^{k_i}{c_{i,j}r_{i,j}}=m,\quad c_{i,j}\in\{0,1\},\\
\text{and for all $(i,j)$,}\quad \sum_{k\neq j}d_{i,\{j,k\}}=c_{i,j}(m-1),\quad d_{i,\{j,k\}}\in\{0,1\}.
\end{gather*}
The last set of constraints consists in adding one variable for each pair of~$(w_{ij},w_{ik})$ and forcing the sum of these variables for fixed~$i$ to be the number of pairs, namely~$m-1$, an idea taken from~\cite{DehghanAssariShah2015}.

\subsubsection{Example}
Consider the equation
\[zy''+(1-6z)y'+(z-3)y=0,\]
with initial condition $y(0)=1$, which specifies a unique power series solution $S(z)=1+3z+13z^2/2+\dotsb$. It has two singular points, at~0 and~$\infty$. The point 0~is regular, with exponents~$0,0$. The point $\infty$~is irregular, with exponential parts $w_\pm=\alpha_\pm z+1/2$, where $\alpha_\pm=-3\pm2\sqrt2$, corresponding to formal solutions $\exp(-\alpha_\pm z)/\sqrt{z}$ at infinity and both generalized exponents are equal to~$1/2$. In the notation above, we have
\begin{gather*}
s=1,\quad \mu_0=\mu_1=z, \quad k_0=k_1=2,\quad  r_{0,1}=r_{0,2}=r_{1,1}=r_{1,2}=1,\\
w_{0,1}=w_+,w_{0,2}=w_-,w_{1,1}=w_{1,2}=0.
\end{gather*}
Looking for a right factor of order~$m$ leads to maximizing~$A\in\mathbb N$ such that
\[A=-m(m-1)-\left(c_{0,1}/2+c_{0,2}/2-m(m-1)/2\right)
+d_{0,\{1,2\}},
\]
with the constraints
\begin{gather*}
c_{0,1}+c_{0,2}=m,\quad 
d_{0,\{1,2\}}=c_{0,1}(m-1)=c_{0,2}(m-1),
\\
c_{0,1},c_{0,2},d_{0,\{1,2\}}\text{ in }\{0,1\}.
\end{gather*}
For the order~$m=1$ of a right factor, the last constraints force $d_{0,\{1,2\}}=0$, $c_{0,1}+c_{0,2}=1$, which makes $A<0$, showing that there is no solution and thus no factorization with a right factor of order~1; the equation is minimal.

\subsubsection{Example} 
We show in more detail the computation for the order~10 differential
equation~\eqref{eq:deq-Apery31} of \cref{example-Ap31}.

There are two singularities: 0 and infinity. The point~0 is regular, with exponents
\[0,0,0,0,1,1,\alpha_1,\alpha_2,\alpha_3,\alpha_4,\]
with $\alpha_i$ the four roots of the irreducible polynomial
\[P_\alpha=29412 x^4  - 246240 x^3  + 764259 x^2  - 1042332 x + 527381.\]
The point~$\infty$ is \emph{irregular}. Its exponential parts are
\[1,1,\beta_i\ (i=1,\dots,4),\gamma_ix+3/2\ (i=1,\dots,4),\]
with $\beta_i$ and $\gamma_i$ roots of the irreducible polynomials
\[P_\beta=29412 x^4  - 342000 x^3  + 1482459 x^2  - 2838258 x + 2024696,\quad
P_\gamma=x^4  + 16 x^3  - 112x^2  + 284 x + 4.
\]
Thus, for this equation, 
\[S_0(\mathcal L)=2+\sum_i{\alpha_i}-45=-\frac{1489}{43},\quad
S_\infty(\mathcal L)=2+\sum_i{\beta_i}+4\frac{3}{2}-45=-\frac{1901}{43},\quad
I_\infty(\mathcal L)=60
\]
and the generalized Fuchs equation reduces to
\[-\frac{1489}{43}-\frac{1901}{43}-30=-90.\]
For lower orders~$m$, the optimization problem to be solved is
\begin{gather*}
\max A\quad{\text{subject to}}\\
\begin{split}
A=-m(m-1)
-\left(c_{0,1}+c_{0,2}+c_{0,\{3,4,5,6\}}\frac{500}{43}
+c_{0,\{7,8,9,10\}}4\frac{3}{2}-\frac{m(m-1)}2\right)\\
-\left(c_{1,5}+c_{1,6}+c_{1,\{7,8,9,10\}}\frac{360}{43}-
\frac{m(m-1)}2\right)\\
+4d_{0,\{1,\{7,8,9,10\}\}}+4d_{0,\{2,\{7,8,9,10\}\}}
+16d_{0,\{\{3,4,5,6\},\{7,8,9,10\}\}}
+6d_{0,\{\{7,8,9,10\},\{7,8,9,10\}\}}
\in\mathbb N
\end{split}\\
\intertext{with constraints}
c_{0,1}+c_{0,2}+4c_{0,\{3,4,5,6\}}+4c_{0,\{7,8,9,10\}}
=c_{1,1}+c_{1,2}+c_{1,3}+c_{1,4}+c_{1,5}+c_{1,6}
+4c_{1,\{7,8,9,10\}}=m\\
d_{0,\{1,2\}}+4d_{0,\{1,\{3,4,5,6\}\}}
+4d_{0,\{1,\{7,8,9,10\}\}}=c_{0,1}(m-1),\\
d_{0,\{1,2\}}+4d_{0,\{2,\{3,4,5,6\}\}}+
4d_{0,\{2,\{7,8,9,10\}\}}=c_{0,2}(m-1),\\
d_{0,\{1,\{3,4,5,6\}\}}+
d_{0,\{2,\{3,4,5,6\}\}}+
3d_{0,\{3,4,5,6\},\{3,4,5,6\}}+
4d_{0,\{\{3,4,5,6\},\{7,8,9,10\}\}}
=(m-1)c_{0,\{3,4,5,6\}},\\
d_{0,\{1,\{7,8,9,10\}\}}+
d_{0,\{2,\{7,8,9,10\}\}}+
4d_{0,\{\{3,4,5,6\},\{7,8,9,10\}\}}+
3d_{0,\{7,8,9,10\},\{7,8,9,10\}}
=(m-1)c_{0,\{7,8,9,10\}},\\
\text{and for all $(i,j,k)$,}\qquad
c_{i,j}\in\{0,1\},\quad d_{i,\{j,k\}}\in\{0,1\}.
\end{gather*}
Integrality of~$A$ forces $c_{0,\{3,4,5,6\}}=c_{1,\{7,8,9,10\}}$.
If they are both equal to~1, 
the first two lines of the constraint on~$A$ give a quantity that is
at most $-20$. Making $A\ge0$ then requires $d_{0,\{3,4,5,6\},
\{7,8,9,10\}}=1$. The last constraint then makes $c_{0,
\{7,8,9,10\}}=1$, which turns the constraint on $A$ into
\[-20-6+16+4d_{0,\{1,\{7,8,9,10\}\}}+4d_{0,\{2,\{7,8,9,10\}\}}
+6d_{0,\{\{7,8,9,10\},\{7,8,9,10\}\}}
\ge0.\]
Therefore $d_{0,\{\{7,8,9,10\},\{7,8,9,10\}\}}=1$ and at least one of
$d_{0,\{1,\{7,8,9,10\}\}}$ and $d_{0,\{2,\{7,8,9,10\}\}}$ is~1 too.
The last constraint then shows that $m=9$ or~$m=10$ depending on
whether one or two of them are~1. We know that $m=10$ is
possible: it is the original equation. If $m=9$, then the first
constraint gives~$c_{0,1}+c_{0,2}=1$. Injecting into the constraint
for~$A$ makes $A<0$, a contradiction.

We have thus proved that for a
strict factor of~$\mathcal A$, $c_{0,\{3,4,5,6\}}=c_{1,
\{7,8,9,10\}}=0$.
This makes all variables~0 in the left-hand side of the penultimate
constraint. The last constraint becomes
\[d_{0,\{1,\{7,8,9,10\}\}}+
d_{0,\{2,\{7,8,9,10\}\}}+
3d_{0,\{7,8,9,10\},\{7,8,9,10\}}
=(m-1)c_{0,\{7,8,9,10\}}.\]
If $c_{0,\{7,8,9,10\}}$ was equal to~0, then the constraint on~$A$
would give $c_{0,1}=c_{0,2}=0$ too, which would give $m=0$ in the
second one,
a contradiction. Therefore $c_{0,\{7,8,9,10\}}=1$.
The remaining constraints are
\begin{gather*}
\begin{split}A=-\left(c_{0,1}+c_{0,2}
+6+c_{1,5}+c_{1,6}\right)\qquad\qquad\qquad\qquad\qquad\qquad\qquad\\
+4d_{0,\{1,\{7,8,9,10\}\}}
+4d_{0,\{2,\{7,8,9,10\}\}}
+6d_{0,\{\{7,8,9,10\},\{7,8,9,10\}\}}\ge0,
\end{split}\\
c_{0,1}+c_{0,2}+4
=c_{1,1}+c_{1,2}+c_{1,3}+c_{1,4}+c_{1,5}+c_{1,6}
=m,\\
d_{0,\{1,2\}}
+4d_{0,\{1,\{7,8,9,10\}\}}=c_{0,1}(m-1),\\
d_{0,\{1,2\}}+
4d_{0,\{2,\{7,8,9,10\}\}}=c_{0,2}(m-1),\\
d_{0,\{1,\{7,8,9,10\}\}}+
d_{0,\{2,\{7,8,9,10\}\}}+
3d_{0,\{7,8,9,10\},\{7,8,9,10\}}
=m-1.
\end{gather*}
The second
one then implies that \emph{the order
of a
strict right factor of~$\mathcal{L}$ can only be one of~$\{4,5,6\}$.}

With $m=6$, there is only one solution (meaning that no optimization
is needed), with all the remaining
variables equal to~1 and the bound $A$ on the number of apparent
singularities equal to~4. There are therefore at most~5 regular
singularities (these four and~0, which is a regular singularity
of~$\mathcal L$). Such a factor can be written
\[\mathcal M=\partial_z^6+\frac{a_5}{A}\partial_z^5+\dots+\frac{a_0}
{A^6},\]
with $A$ of degree at most~5. The Newton polygon of~$\mathcal L$ at
infinity has for vertices $(0,0),(6,0),(10,4)$. The largest
possibility for~$\mathcal M$ is therefore~$(0,0),(2,0),(6,4)$, leading
to the following bounds on the degree of the numerators~$a_i$: $
(\deg A-1,\deg A^2-2,\deg A^3-4,\deg A^4-4,\deg A^5-4,\deg A^6-4)$.
Reducing to the same denominator gives the bounds $
(30,29,28,26,26,26,26)$ on the degrees of the coefficients of~$
(\partial^6,\dots,1)$. This is the bound used in Example%
~\ref{example-Ap31},
leading to the discovery of the factor~$\mathcal M$.

With $m=5$, there are several solutions, which are as in the
case
when $m=6$, but with one of~$c_{0,1}$ and~$c_{0,2}$ equal to~0 and
consequently $d_{0,\{1,2\}}=0$ and one of $d_{0,\{1,\{7,8,9,10\}\}}$,
$d_{0,\{2,\{7,8,9,10\}\}}$ equals~0, leading to a bound~$A\le 6+4-
(6+1+1)=2$ on the number of apparent singularities. By the same
reasoning as above, this leads to the bounds $(15,14,13,13,13,13)$ on
the degrees of the coefficients of a factor of order~5. Using about
90~coefficients of the series shows that such a factor does not exist.

Finally, with~$m=4$ the only of the remaining~$d$ variables that
is not~0 is $d_{0,\{7,8,9,10\},\{7,8,9,10\}}$ and the bound on~$A$
becomes $6-(6)=0$. Again, a computation with degree bounds $
(5,4,3,3,3)$ proves that no factor of degree~4 exists.

\subsubsection*{Note on the relaxed problem.} For efficiency reasons,
one may prefer solving the relaxed optimization problem where the
variables $c_
{i,j}$
and $d_{i,\{j,k\}}$ are not restricted to the set~$\{0,1\}$, but
can be real numbers in the interval~$[0,1]$. Also, 
during the optimization, $A$ is
not restricted to be an integer. Then what happens in this example is
that the absence of a factor of order~9 is not detected. Instead, the
optimization finds a solution with $A=4.5$, which leads to a bound
of~4 on the number of apparent singularities and thus of~45 on the
degree of the coefficients. With this bound, sufficiently many
coefficients are computed so that the computation of an approximant
basis finds a
candidate operator. This turns out to be the factor~$\mathcal M$ of
order~6 above. The next stage is to prove its minimality. For order~5,
the optimization of the relaxed problem
gives a bound equal to~5 for~$A$ (to be compared with~2 above),
leading to a bound of~30 on the degree of the coefficients. Using
about~190 coefficients of the series (instead of 90~above) shows that
no such factor exists. For lower orders, the relaxed problem does not
have any solution, concluding the computation.

\subsection{Inhomogeneous minimization} \label{ssec:nhm}
Again, we consider an equation like \cref{eq:defL} and initial conditions for a unique formal power series~$S$ solution of it. Using the method of the previous section, we can assume that it has minimal order. The problem of inhomogeneous minimization is to find an equation
\[\mathcal M(y(z))=B(z),\quad\text{with}\quad \mathcal M=b_s(z)\partial_z^s+\dots+b_0(z),\]
with $s<r$ and rational function coefficients~$b_0,\dots,b_s,B$ ($b_s\neq0$), having~$S$ as a solution. When such an equation exists with~$B\neq0$, applying $B\partial_z-B'$ on both sides of the equation yields a homogeneous linear differential equation of order~$s+1$ satisfied by~$S$, so that 
minimality of~$\mathcal L$ implies $s=r-1$. Without loss of generality (up to replacing $\mathcal M$ by $\frac{1}{B} \mathcal M$) one can assume~$B(z)=1$ and then differentiation implies \[\partial\mathcal M=R(z)\mathcal L\]
for some non-zero rational function~$R$, which is therefore an \emph{integrating factor} of~$\mathcal L$. This implies that~$R$ is a rational function solution of the adjoint equation~\cite[Chap.~III.\S10]{Poole1960}
\begin{equation}\label{eq:adjoint}
\mathcal L^*(R)=0.
\end{equation}
Finding rational solutions of linear differential equations is a classical problem, whose solution can be found by an algorithm due to Abramov~\cite{Abramov89,AbKv91}, with roots in Liouville's work~\cite{Liouville1833}. This algorithm returns a basis of rational solutions of \cref{eq:adjoint}. This is a decision algorithm: if no non-zero rational solution is found, this proves that there is no inhomogeneous linear differential equation of order smaller than~$r$ satisfied by the power series~$S$. Otherwise, minimality implies that the basis consists of one solution~$R(z)$. The operator~$\mathcal M$ (known as the \emph{bilinear concomitant}~\cite{Poole1960}) can be reconstructed coefficient by coefficient (this is equivalent to~\cite[p.~703]{AdamczewskiRivoal2018}). Then by design, $\mathcal M(S)$ is a constant~$c$, which can be computed from the first coefficients of the power series~$S$, completing the computation of the minimal inhomogeneous equation~$\mathcal M(y)=c$ satisfied by~$S$. This computation is summarized in \cref{algo:min-inhom}.

\begin{algorithm}
\caption[]{Minimal inhomogeneous linear differential equation}\label{algo:min-inhom}
\begin{algorithmic}[1]
\Require{\ \quad$\mathcal L=a_r(z)\partial_z^r+\dots+a_0(z)$,\par
\qquad a linear operator of minimal order that vanishes at~$S$;\par
\quad ini: $S_0$ a truncated power series at precision $\ge \max\mathcal Z_\mathcal L$
}
\Ensure{\ \ $\mathcal M=b_{r-1}(z)\partial_z^{r-1}+\dots+b_0(z)$ and $B(z)\in\mathbb Q(z)$,
such that~$\mathcal M(S)=B$\par
\quad or {\sc Fail} if no such pair exists.}
\State $\mathcal L^\star\coloneq \operatorname{adjoint}(\mathcal L)$;
\State $\mathcal S=\Call{BasisRationalSolutions}{\mathcal L^\star}$
\If{$\mathcal S=\emptyset$}{\ \Return{FAIL}}\EndIf
\State Let $R$ be the unique element of $\mathcal S$
\State $b_{r-1}\coloneq R a_r$
\For{$j=r-2,\dots,0$}
       $b_j\coloneq R a_{j+1}-b_{j+1}'$
\EndFor
\State Compute $S$ up to precision $r-\min_j\operatorname{val}_0(b_j)$
\State Let $B$ be the constant term of $\mathcal M(S)$
\State\Return{$\mathcal M,B$}
\end{algorithmic}
\end{algorithm}


\section{Efficient computation of the set of algebraic values taken by  \texorpdfstring{$E$}{E}-functions at algebraic points}\label{sec:E-functions}

\subsection{The Adamczewski-Rivoal algorithm}\label{ssec:AR}
In this section, we consider mainly $E$-functions with Taylor coefficients in~$\mathbb Q$. 
An \emph{$E$-function over~$\mathbb Q$} is a power series  
\[f(z)\coloneq \sum_{n=0}^{\infty} \frac{a_n}{n!} z^n\text{ in }\Q[[z]]\]
with $a_n\in\mathbb Q$ and such that
there exists $C>0$ 
with the following properties:
\begin{enumerate}
\item[$(i)$] $f$ satisfies a homogeneous linear differential equation with 
coefficients in~$\Q(z)$;
\item[$(ii)$] for any $n\ge 0$, $\vert a_n\vert \le C^{n+1}$;
\item[$(iii)$] for any $n\ge 0$, there exists $d_n \in \mathbb N\setminus \{ 0\} $ such that $d_n  \leq C^{n+1}$ and 
$d_na_m\in \Z$ for all~$0\le m\le n$.
\end{enumerate}
We shall sometimes simply write ``LDE'' for ``linear differential equation with 
coefficients in~$\Q(z)$ or in $\Qbar(z)$''; unless otherwise stated, an LDE will be assumed to be homogeneous.  
In the rest of this section, $E$-functions over $\mathbb Q$ are simply called $E$-functions. This is justified by the fact that most of the discussion applies to more general settings, in particular to $E$-functions with Taylor coefficients in $\Qbar$ and to $E$-functions in Siegel's more general sense, as discussed in \cref{ssec:extensions1,ssec:extensions2}.

$E$-functions are entire functions (by $(ii)$). 
Polynomials in $\Q[z]$ are trivial examples of $E$-functions; all non-polynomial $E$-functions are transcendental over $\Q(z)$.
The class of $E$-functions includes the exponential function $\exp(z)$, Bessel's function of the first kind
\[J_0(z) \coloneq  \sum_{m = 0}^\infty \frac{\left(-1\right)^{m}}{m !^{2}} \left(\frac{z}{2}\right)^{2 m}= \, _0 F_1 [\,\cdot \, ; 1 ; -z^2/4],\]
and more generally the {\em hypergeometric $E$-functions}, i.e. series of the form 
\[
{}_pF_q[a_1, \ldots, a_p;b_1, \ldots, b_q;\lambda z^{q-p+1}] \coloneq \sum_{n=0}^\infty \frac{(a_1)_n\cdots (a_p)_n}{(1)_n(b_1)_n\cdots (b_q)_{n}}\lambda^n z^{(q-p+1)n}
\] 
with rational parameters~$a_i, b_j$, $q\ge p\ge 0$, $\lambda \in \Qbar^*$ and where $(\alpha)_0\coloneq 1$,   $(\alpha)_n\coloneq \alpha(\alpha+1)\cdots (\alpha+n-1)$ for $n\ge 1$. $E$-functions form a sub-ring of the ring of formal power series in $\Q[[z]]$, stable by $d/dz$ and $\int_0^z$; these properties can be used to construct many examples of $E$-functions starting from hypergeometric series. Shidlovskii has proved in \cite[p.~184]{shidl} that any $E$-function solution of an LDE of order $1$ is of the form $p(z)e^{\lambda z}$ for some $p(z)\in \Qbar[z]$ and $\lambda \in \Qbar$. Gorelov has proved in \cite{gorelov2} that any $E$-function solution of an LDE of order $2$ is a $\Qbar(z)$-linear combination of hypergeometric $E$-functions with $p=q=1$ (he had obtained earlier in \cite{gorelov1} a similar but more precise result for $E$-functions solution of an inhomogeneous LDE of order $1$). 
However, Fres\'an and Jossen have recently showed in~\cite{FrJo21} that not all $E$-functions are $\Qbar(z)$-linear combinations of hypergeometric $E$-functions, nor even more generally polynomials in hypergeometric $E$-functions with algebraic coefficients.

As of today, 
no algorithm is known neither for deciding whether a linear differential equation 
$\mathcal L(y(z))=0$ admits solutions that are $E$-functions, nor for deciding whether 
a solution of $y(z)$ of~$\mathcal L$, uniquely determined by sufficiently many initial conditions, is  an $E$-function.
It is actually not clear whether these questions are decidable or not.
Consequently, the algorithm described below relies on the following assumption:
\begin{enumerate}
\item[(A)] An oracle guarantees that the input $f$ is an $E$-function.
\end{enumerate}
In practice, an $E$-function is given by an explicit  expression for its Taylor coefficients as a multiple hypergeometric sum and $\mathcal L$ 
can then be computed for instance by Zeilberger's creative telescoping algorithm~\cite{Zeilberger91}.

Siegel initiated a program to determine when an $E$-function takes a
transcendental value at an algebraic point~\cite{siegel}. This culminated with
the celebrated Siegel-Shidlovskii theorem: given a vector $Y$ of $E$-functions
$f_1, \ldots, f_n$ solution of a differential system $Y'=AY$ with a
matrix $A$ with elements in $\Qbar(z)$, the transcendence degree over
$\Qbar(z)$ of the field generated by $f_1(z), \ldots, f_n(z)$ over $\Qbar(z)$
is equal to the transcendence degree over $\Qbar$ of the field generated by
$f_1(\alpha), \ldots, f_n(\alpha)$ over $\Qbar$ for every non-zero algebraic
number $\alpha$ which is not a singularity of $A$ (i.e., a pole of some
element of $A$). In 2006, Beukers~\cite[Thm.~1.3]{Beukers06} refined this theorem by
proving that any homogeneous polynomial relation between the values
$f_1(\alpha), \ldots, f_n(\alpha)$ with coefficients in $\Qbar$ is a
specialization of a homogeneous polynomial relation between the functions
$f_1(z), \ldots, f_n(z)$ with coefficients in $\Qbar(z)$, again when $\alpha$
is not a singularity of $A$. A less precise version of this theorem (but for
$E$-functions in Siegel's more general sense; see \cref{ssec:extensions2} for
details) had been obtained in 1996 by Nesterenko and Shidlovskii \cite{NeS96},
where $\alpha$ is simply assumed not to lie in a certain finite set $S$,
depending on the $f_j$'s but not specified in their article. A fundamental
consequence of their result is that a transcendental $E$-function $f$ takes
only finitely many algebraic values when evaluated at algebraic points. To see
this, one considers a non-trivial minimal inhomogeneous differential equation
with polynomial coefficients
$p+\sum_{j=0}^\mu p_j f^{(j)}=0$ satisfied by $f$ over $\Qbar(z)$ and
applies the Nesterenko-Shidlovskii theorem to the functions $f_1\coloneq 1$, $f_2\coloneq f,
\ldots, f_{\mu}\coloneq f^{(\mu-1)}$. They are linearly independent over $\Qbar(z)$,
hence the numbers $1, f(\alpha), \ldots, f^{(\mu-1)}(\alpha)$ are
$\Qbar$-linearly independent over $\Qbar$ for all $\alpha \in \Qbar\setminus
S$; in particular $f(\alpha) \notin \Qbar$ for all such $\alpha$'s. For
$E$-functions in the strict sense, we now know thanks to Beukers~\cite{Beukers06} that if
$\alpha\in \Qbar^*$ is not a root of the leading coefficient $p_{\mu}$ above,
then $f(\alpha)\notin \Qbar$.

\begin{algorithm}[t]
\caption[]{Algebraic values of $E$-functions over~$\mathbb Q$}\label{algo:alg-values}
\begin{algorithmic}[0]
\Require{$\mathcal L=a_r(z)\partial_z^r+\dots+a_0(z)$;\\
\phantom{Input}ini: $f_0$ a truncated power series at precision $p_0\ge r$\\ \qquad\qquad specifying a unique solution~$f\in\mathbb Q[[z]]$ of $\mathcal L(f)=0$.\\
\phantom{Input}\emph{It is assumed that $f$ is an $E$-function.}}
\Ensure{Either ``$f$ is a polynomial'',\\ 
\qquad\qquad or the finite set of all identities $f(\alpha)=\beta$ with algebraic $\alpha$ and $\beta$.}
\State $\mathcal L_{\min}\coloneq $\Call{MinimalRightFactor}{$\mathcal L$,ini}\Comment{\Cref{algo:min-hom}}
\State $\mathcal L_{\text{inhom}},g\coloneq $\Call{MinimalInhomogeneousRightFactor}{$\mathcal L_{\min}$}\Comment{\Cref{algo:min-inhom}}
\If{$\ord\mathcal L_{\text{inhom}}=0$}{ \Return{$f$ is a polynomial}}\EndIf
\State Define the polynomials $v_0,\dots,v_{s+1}$ by $\mathcal L_{\text{inhom}}(f)-g=v_0f^{(s)}-v_1f^{(s-1)}-\dots-v_sf-v_{s+1}$
\State Form the companion matrix~$M$ s.t. $(0, f', f'',\dots,f^{(s)})^{\mathsf T}=M\cdot (1,f,f',\dots,f^{(s-1)})^{\mathsf T}$
\State $\mathcal R\coloneq \{f(0)=f_0(0)\}$
\For{$\mu\in\Q[z]$ irreducible factor of~$v_0$ with $\mu(0)\neq0$}
\State Write $\alpha$ for a root of~$\mu$
\State $B\coloneq $\Call{BeukersAlgo}{$M,\alpha$}
\State $(b_1,\dots,b_m)\coloneq $basis of the left kernel of $B(\alpha)$\Comment{Basis of algebraic relations at~$\alpha$}
\If{there exists $(\beta,-1,0,\dots,0)$ in $\operatorname{Span}(b_1,\dots,b_m)$}
\State $\mathcal R\coloneq \mathcal R\cup\{f(\alpha)=\beta\}$
\EndIf
\EndFor
\Return $\mathcal R$
\end{algorithmic}
\end{algorithm}

Thus, in order to completely determine when an $E$-function takes a transcendental value at a given non-zero algebraic point, one issue needs to be dealt with: what happens for the (finite number of) algebraic numbers that are roots of $p_{\mu}$ (in the same setting as above). This was done by Adamczewski and Rivoal~\cite{AdamczewskiRivoal2018} by pushing further Beukers' ideas from~\cite{Beukers06}. The end result is an algorithm that takes as input an $E$-function~$f$ and
either detects that $f$ is algebraic (in which case it is even a polynomial),
or computes the (finite) list of identities $f(\alpha)=\beta$ for algebraic values $\alpha$ and~$\beta$.

\Cref{algo:alg-values} gives a version of the algorithm by Adamczewski and Rivoal that benefits from the minimization algorithms of \cref{sec:minimization}. It is stated here for the $E$-functions over~$\mathbb Q$ considered in this section (see the comments in \cref{ssec:extensions1,ssec:extensions2} below for its extension to more general settings).
The algorithm relies on two results due to Beukers~\cite{Beukers06}: 
\begin{enumerate}
\item If $F=(f_1,\dots,f_n)^{\mathrm T}$ with $E$-functions~$f_i$ is a solution of~$Y'=AY$, the entries of $A$ belonging to~$\Qbar(z)$ and if~$f_1(z),\dots,f_n(z)$ are linearly independent over~$\Qbar(z)$, then for any 
non-zero~$\alpha$ that is not a pole of~$A$, the numbers~$f_1(\alpha),\dots,f_n(\alpha)$ are linearly independent over~$\Qbar$~\cite[Corollary 1.4]{Beukers06};
\item Under the same assumptions, there exists a matrix~$M$ with entries in~$\Qbar[z]$ such that~$F=ME$, and $E$ is a vector of $E$-functions solution to a system $Y'=BY$ where $B$ does not have any non-zero pole~\cite[Theorem 1.5]{Beukers06}.\label{item-2}
\end{enumerate} 
Starting from a linear differential operator $\mathcal L$ and initial conditions specifying an $E$-function~$f$, the algorithm first computes a minimal inhomogeneous equation of order~$s$ for~$f$. (This step also allows to detect and discard a polynomial~$f$.) By  minimality of this equation, $F=(1,f,f',\dots,f^{(s-1)})$ is a vector of $\Qbar(z)$-linearly independent $E$-functions solution to a matrix deduced from the equation. Given a matrix~$M$ as in the result~(\ref{item-2}) above, it follows that the points~$\alpha$ where $(1,f(\alpha),\dots,f^{(s-1)}(\alpha))$ are linearly dependent over~$\Qbar$ are non-zero poles of~$A$ where the left-kernel of~$M$ is not reduced to~0.
The specific case of $f(\alpha)$ being algebraic corresponds to the existence of a non-zero vector in that kernel whose first two coordinates only are not zero.
The remaining question is the computation of these matrices~$M$, which is described in~\cref{ssec:beukersalgo}.

Note that the first two steps of \Cref{algo:alg-values}, i.e. the calls to
\Call{MinimalRightFactor}{$\mathcal L$,ini} and
\Call{MinimalInhomogeneousRightFactor}{$\mathcal L_{\min}$}, are not specific
to $E$-functions, and both return an output even when $f$ is not an
$E$-function. In that case, \Call{BeukersAlgo}{$M,\alpha$} terminates (by
design) and it may even output $\alpha$'s such that $f(\alpha)\in \Qbar$; but
it can no longer be claimed that all such $\alpha$'s have been found.

\subsection{Beukers' algorithm and desingularization}\label{ssec:beukersalgo}

Algorithm \ref{algo:alg-values} concludes with a call to Algorithm \Call{BeukersAlgo}{$M,\alpha$} described below.
It is a clever desingularization process, 
which is different from the one developed by Barkatou and Maddah~\cite{BaMa15}, in that 
it does not rely on Moser's reduction~\cite{Moser60,Barkatou95}. The end result is the following (Theorem 1.5 in~\cite{Beukers06}).

\begin{theorem} \label{theo:beukersdesing}
        Let $Y = (f_1, \ldots, f_n)^{\mathrm T}$ be a vector of $\overline{\Q}(z)$-linearly independent $E$-functions satisfying $Y' = AY$, where $A$ is an $n \times n$ matrix with entries in~$\overline{\Q}(z)$. Then, there exists a vector of $E$-functions $Z=(e_1,\ldots, e_n)^{\mathrm T}$ solution of $Z'=BZ$ with $B$ having entries in  $\overline{\Q}[z,1/z]$, and there exists a polynomial matrix $M$ with entries in  $\overline{\Q}[z]$ and $\det(M)\neq 0$, such that $(f_1, \ldots, f_n)^{\mathrm T} = M \cdot (e_1, \ldots, e_n)^{\mathrm T}$.
\end{theorem}

The key properties used in the proof of \cref{theo:beukersdesing} are 
the statements {\bf (P1)}, {\bf (P1')} and {\bf (P2)} listed below.
Note that for an $E$-function $f\in\Q[[z]]$, we write $L^{\textrm{min}}_{f}$ for  the monic linear differential operator in $\Q(z)\langle\partial_z\rangle$ of minimal order that cancels~$f$.

\medskip 

{\bf (P1)} For any $E$-function $f\in \Q[[z]]$, the finite non-zero singularities of $L^{\textrm{min}}_{f}$ are apparent; 

\smallskip {\bf (P1')} 
If $A$ is an $n \times n$ matrix with entries in~$\overline{\Q}(z)$,
and if $F$ is a vector of $\overline{\Q}(z)$-linearly independent $E$-functions satisfying $F'=AF$,
then the finite non-zero singularities of the system $Y'=AY$ are apparent; 

\smallskip {\bf (P2)} If an $E$-function $f$ and $\alpha\in \Qbar$ are such that $f(\alpha)\in \Qbar$, then $({f(z)-f(\alpha)})/({z-\alpha})$ is an $E$-function. 

\smallskip

Property {\bf (P1)} is \emph{Andr\'e's theorem}~\cite[Cor.~4.4]{Andre00}
and {\bf (P2)} is an important property of $E$-functions proved by Beukers~\cite[Prop.~4.1]{Beukers06}.
Property {\bf (P1')} is a system version of Andr\'e's theorem, which is not, to our knowledge, explicitly stated in the literature, although it is implicitly contained in Beukers' proof of his Theorem 1.5 in~\cite[p.~378]{Beukers06}. 
For completeness, we detail the proof of {\bf (P1')}, which goes along the following lines.

\medskip
\emph{Proof of {\bf (P1')}}.
Let $\mathcal{G}$ be the differential Galois group of (the Picard-Vessiot field of) $Y'=AY$.
Let $V$ 
be the $\Qbar$-vector space
generated by the orbit $\{ \sigma ( F ) \mid \sigma \in \mathcal{G}\} $, 
where $F = (f_1, \ldots, f_n)^{\mathrm T}$ is a vector of $E$-functions satisfying $F'=AF$, linearly independent over~$\Qbar(z)$.
The conclusion of {\bf (P1')} clearly follows by combining the following two steps.

{\bf Step 1}. The dimension of $V$ over $\Qbar$ is equal to $n$, hence one can extract from $V$ a fundamental matrix of solutions $\mathcal{F}$, whose first column is $F$.

{\bf Step 2}. All the entries of $\mathcal{F}$ are holomorphic at all non-zero points $\alpha \in \mathbb{C}\setminus \{ 0\}$.

\medskip
\emph{Proof of {\bf Step 2}}: Let $L_i \coloneq  L^{\textrm{min}}_{f_i}$. Since elements of $\mathcal{G}$ commute with differentiation, all $\sigma(f_i)$ are solutions of $L_i$ for all $\sigma \in \mathcal{G}$.
By {Andr\'e's theorem} {\bf (P1)}, $\sigma(f_i)$ has no true singularity in $\mathbb{C}\setminus \{ 0\}$. Hence $\mathcal{F}$ is holomorphic at any $\alpha \in \mathbb{C}\setminus \{ 0\}$. 

\medskip
\emph{Proof of {\bf Step 1}}: 
If $A$ is a companion matrix, then the shape of the system $Y'=AY$ implies that $F$
is of the form $F=(f, f', \ldots, f^{(n-1)})^{\mathrm T}$, where $f$ is the $E$-function $f=f_1$.
The linear independence assumption implies that 
$L^{\textrm{min}}_{f}$ has order $n$.
By~\cite[Corollary 1.38]{PuSi03}
(see also~\cite[p.~190, Proposition~3]{BertrandBeukers1985}),
the dimension of the 
$\Qbar$-vector space~$\tilde V$ generated by the orbit $\{ \sigma (f) \mid \sigma \in \mathcal{G}\} $
is equal to~$n$.
On the one hand, this dimension is upper bounded by the dimension of $V$,
since any linear relation among the entries of $V$ yields a linear relation among the elements of the set $\{ \sigma ( f ) \mid \sigma \in \mathcal{G}\} $.
On the other hand, $V$ is included in the solution space of $Y'=AY$, hence it has dimension at most~$n$.
Therefore, $\dim_{\Qbar}(V)=n$, and the assertion is proved in the companion case.

Now, if $A$ is a general matrix, by the cyclic vector lemma (see
e.g.~\cite[Thm~3.11]{ChKo02},
or~\cite[Proposition 2.9]{PuSi03}) the system $Y'=AY$ is ``gauge equivalent'' to $Z'=CZ$, where~$C$ is a companion matrix with entries in~$\Qbar(z)$. 
This means that there exists an invertible matrix~$P$ with entries in~$\Qbar(z)$
such that $Z\coloneq P \cdot Y$ satisfies $Z'=P[A] \cdot Z$,
where $P[A]\coloneq (PA+P')P^{-1}$ is equal to a companion matrix~$C$.
Moreover, 
by~\cite[\S 6]{Cope36} (see also~\cite[Lemma 2.10]{PuSi03}),
the entries of the matrix $P$ can be chosen to be polynomials in $\Qbar[z]$, of degree at most $n-1$.
By construction, $G\coloneq P\cdot F$ satisfies $G'= C \cdot G$.
Hence, the vector $G$ is necessarily of the form
$G\coloneq (g,g',\ldots,g^{(n-1)})^{\mathrm T}$, where $g$ is a $\Qbar[z]$-linear combination of the $E$-functions~$f_i$.
In particular, $g$ is itself an $E$-function.
Moreover, $L^{\textrm{min}}_{g}$ has order $n$: indeed, any $\Qbar[z]$-linear combination $0 = \mathbf{v} \cdot G$ between $g,g',\ldots,g^{(n-1)}$ yields a $\Qbar[z]$-linear combination $0 = (\mathbf{v}P) \cdot F$ between the entries of $F$;
since these are assumed linearly independent over $\Qbar(z)$, and since $P$ is invertible, $\mathbf{v}$ is necessarily zero.
We are now in position to apply the companion case.
Since gauge equivalent systems have the same differential Galois group~\cite[p.~13]{Singer09}, the new companion system $Z'=CZ$ has differential Galois group~$\mathcal{G}$. 
By applying the companion case, we deduce that $\dim_{\Qbar}(V_C)=n$, where 
$V_C$ is the $\Qbar$-vector space generated by the orbit 
$\{ \sigma (G) \, | \, \sigma \in \mathcal{G}\} $.
It remains to show that $\dim_{\Qbar}(V) = \dim_{\Qbar}(V_C)$. 
Choose $\sigma_1, \ldots, \sigma_n$ in $\mathcal{G}$ such that $\sigma_1(G), \ldots, \sigma_n(G)$ 
are linearly independent over~$\Qbar$.
Then, $\sigma_1(F), \ldots, \sigma_n(F)$ 
are also linearly independent over~$\Qbar$, because of the relation $G=P \cdot F$ and the fact that all elements in $\mathcal{G}$ leave $P$ invariant.
It follows that $V$ has dimension at least~$n$; since $V$ is included in the solution space of $Y'=AY$, it also has dimension at most~$n$, therefore $\dim_{\Qbar}(V) =n$, which concludes the proof.
\hfill  $\square$

\medskip
We now prove Theorem \ref{theo:beukersdesing}.
Our proof is inspired by that of Beukers in~\cite{Beukers06}.
The main difference is that our proof does not depend on a specific desingularization procedure for linear differential systems.

\begin{proof}[Proof of Theorem \ref{theo:beukersdesing}] 

We make use of a \emph{desingularization lemma}~\cite[Theorem~2]{BaMa15}:
there exists a polynomial matrix $M$ with entries in  $\overline{\Q}[z]$ and with $\det(M) \neq 0$ such that the finite poles of $B = M[A]\coloneq M^{-1} (AM-M')$ are exactly the true (i.e, non-apparent) singularities of~$Y'=AY$
and such that $\det(M)$ is a non-zero polynomial whose roots are among the apparent singularities of~$Y'=AY$. (See also \cref{prop:algo-beukers} below.)

In our case, by Property {\bf (P1')} above, the entries of $B$ are in  $\overline{\Q}[z,1/z]$.

Define $ Z = (e_1, \ldots, e_n)^{\mathrm T} \coloneq  M^{-1} \cdot  (f_1, \ldots, f_n)^{\mathrm T}$, so that  
 $(f_1, \ldots, f_n)^{\mathrm T} = M \cdot (e_1, \ldots, e_n)^{\mathrm T}$. A simple computation shows that  $Z'=BZ$. 
 It remains to prove that all the $e_i$'s are $E$-functions.
 The proof relies on Property~{\bf (P2)} above.
 
 By definition, each $e_i$ is equal to $\frac{1}{\det (M)} \cdot \sum_{j=1}^n p_{i,j} f_j$ for polynomials $p_{i,j}$ in~$\overline{\Q}[z]$. Since $B$ has no non-zero pole, each $e_i$ is holomorphic 
 at every apparent singularity $\rho \neq 0$ of $Y'=AY$. Therefore, 
 $ \sum_{j=1}^n p_{i,j} f_j$ is an $E$-function which vanishes at any  root of $\det(M)$ at an order at least equal to the multiplicity of that root in $\det(M)$. By repeated application of Property~{\bf (P2)}, it follows that $e_i$ is an $E$-function.
\end{proof}
Beukers' proof~\cite[p.~378]{Beukers06} of Theorem~\ref{theo:beukersdesing} actually contains a general effective desingularization process, which deserves to be stated independently
of the context of $E$-functions.
It is given in Algorithm~\ref{algo:Beukers}, whose properties are summarized in the following.
\begin{proposition}\label{prop:algo-beukers}
Let $A$ be an $n\times n$ matrix with entries in ${\mathbb Q}(z)$ and let $\alpha\in\overline{\mathbb Q}$ be such that a fundamental solution~$\mathcal Y$ of $Y'=AY$ is holomorphic at~$\alpha$. Then \cref{algo:Beukers} computes a matrix of polynomials~$B\in(\mathbb Q(\alpha)[z])^{n\times n}$ such that 
$\mathcal Y=B\mathcal Z$ with $\mathcal Z$ a fundamental solution of $Z'=CZ$ also holomorphic at~$\alpha$ and $C\in(\mathbb Q(\alpha)(z))^{n\times n}$ only has poles where~$A$ does, except at $\alpha$, where it is holomorphic.
\end{proposition}
\begin{proof}
We reproduce Beukers' proof, with more details.

By hypothesis, the determinant $W=\det\mathcal Y$ is holomorphic in a neighborhood of~$z=\alpha$. If $W(\alpha)\neq0$, then $\mathcal Y^{-1}$ is holomorphic in the neighborhood of $z=\alpha$ and therefore so is $A=\mathcal Y'\mathcal Y^{-1}$. In that case, $C=A$ and $B=\Id_n$ gives the result. 

Otherwise, as $W\neq0$, there exists~$r\in\mathbb N_{>0}$ such that $W(z)\sim c(z-\alpha)^r$ for $z\rightarrow\alpha$ with $c\neq0$. Since $W$ satisfies $W'=\operatorname{Trace}(A)W$, it follows that~$r$ is the residue of $\operatorname{Trace}(A)$ at $z=\alpha$. Starting with $B=\Id_n$, $C=A$ and $\mathcal Z=\mathcal Y$, the algorithm repeats at most~$r$ times an operation that updates~$C$ and~$B$ so that $\mathcal Y=B\mathcal Z$ and
\begin{itemize}
    \item[---] $B$ is a matrix of polynomials;
    \item[---] $\mathcal Z$ is holomorphic at $\alpha$;
    \item[---] $C\coloneq B^{-1}(AB-B')$ has no pole outside those of $A$;
    \item[---] $\mathcal Z$ is a fundamental solution of $Z'=CZ$;
    \item[---] $\operatorname{val}_{z=\alpha}\det\mathcal Z=\operatorname{val}_{z=\alpha}\det\mathcal Y-1$.
\end{itemize}
By composing these steps, it is sufficient to prove that one iteration of the loop has these properties.

Each step is centered around the definition of a matrix~$M$ as follows.
Let $k>0$ be the order of the pole of the matrix~$C$ at~$\alpha$, 
let $i$ be the index of the first row of~$C$ with a pole of order~$k$, 
let $v$ be the constant vector of coefficients of~$(z-\alpha)^{-k}$ in that row, $D$ be the diagonal matrix $\operatorname{diag}(1,\dots,1,z-\alpha,1,\dots,1)$ with $z-\alpha$ in the $i$th position and $M$ be an invertible constant matrix with~$v$ in its $i$th row.
Then $\tilde{B}\coloneq BM^{-1}D$ possesses the desired properties:
\begin{itemize}
    \item[---] $\tilde B$ is the product of matrices of polynomials;
    \item[---] $\tilde{\mathcal Z}\coloneq D^{-1}M\mathcal Z$ is holomorphic at $\alpha$: since both $\mathcal Z$ and $\mathcal Z'$ are holomorphic at~$\alpha$, the product~$v\mathcal Z$ is~0 at~$\alpha$, making the $i$th row of $M\mathcal Z$ a multiple of $(z-\alpha)$;
    \item[---] $\tilde{C}\coloneq D^{-1}M(CM^{-1}D-M^{-1}D')$ has no pole outside those of $C$;
    \item[---] $\tilde{\mathcal Z}$ is a fundamental solution of $Z'=\tilde{C}Z$;
    \item[---] $\det\tilde{\mathcal Z}=\det M\det\mathcal Z/(z-\alpha).$\qedhere
\end{itemize}
\end{proof}

\begin{algorithm}
\caption[]{Removal of Singularities (\Call{BeukersAlgo}{$M,\alpha$})}\label{algo:Beukers}
\begin{algorithmic}[0]
\Require{$A$: matrix in~$\Q(z)^{n\times n}$;\\
\phantom{Input}$\alpha$: root of the denominator of an entry in~$A$}
\Ensure{$B$: matrix in $(\Q(\alpha)[z])^{n\times n}$ satisfying \cref{prop:algo-beukers}.}
\State $r\coloneq \operatorname{Res}_{z=\alpha}\operatorname{Trace}(A)$\Comment{Residue of the trace}
\If{$r\not\in\mathbb{N}_{\ge0}$}{ \textbf{error} singularity cannot be removed}
\EndIf
\State $C\coloneq A$; $B\coloneq \Id_n$
\For{$m=1,\dots,r$}
\State $k\coloneq $order of the pole of~$C$ at~$\alpha$
\If{$k=0$}{ \textbf{break}}\EndIf
\State $i\coloneq $index of the first row of $C$ with a pole of order~$k$ at~$\alpha$
\State $v\coloneq $vector of coefficients of $(z-\alpha)^{-k}$ in row~$i$ of~$C$
\State $M\coloneq $an invertible constant matrix with $v$ its $i$th row; \Comment{Complete $v$ into a basis}
\State $D\coloneq \operatorname{diag}(1,\dots,1,z-\alpha,1,\dots,1)$ with $z-\alpha$ in the $i$th position
\State $B\coloneq BM^{-1}D$
\State $C\coloneq D^{-1}MCM^{-1}D-D^{-1}D'$
\EndFor\\
\Return $B$
\end{algorithmic}
\end{algorithm}

When this algorithm is applied in \cref{algo:alg-values} and 
the value of~$r$ computed in the first step of \cref{algo:Beukers}
is~1, which occurs frequently in practice, then the kernel computed in
\cref{algo:alg-values} has dimension~1 and its entries are nothing but
the values of the coefficients of the differential equation
at~$z=\alpha$. In other words, in this situation, there is only one
algebraic relation, which is obtained by evaluating the differential
equation at~$z=\alpha$. It is always possible to obtain an algebraic
relation that way; the strength of \cref{algo:alg-values} is that it
returns a \emph{basis} of all these relations, even when~$r>1$.

\subsection{Effective decomposition of \texorpdfstring{$E$}{E}-functions} \label{ssec:decompositionfpqg}

For an $E$-function $f\in\Q[[z]]$ (or more generally in $\Qbar[[z]]$), we call \emph{exceptional values} those
(finitely many) non-zero algebraic numbers $\alpha$ such that $f(\alpha) \in
\Qbar$. The set of exceptional values of $f$ is denoted by $\textrm{Exc}(f)$.
We call the $E$-function $f$ \emph{purely transcendental} if it has no
exceptional values, i.e. if $\textrm{Exc}(f) =\emptyset$.

This subsection deals with the fact that every $E$-function is equal to the
sum of a polynomial and of a polynomial multiple of a purely transcendental
$E$-function. 
The existence of such a decomposition was proved in~\cite{rivoalnote2016} for
general $E$-functions. We give an alternative proof in the
special case of $E$-functions with coefficients in~$\mathbb Q$
in \cref{theo:decomp}. Before that, in \cref{prop:unicitydecomp}, we
define a canonical
decomposition in the
general case. Decompositions are not unique in general.
Indeed, if we
have $f=p+qg$ where $p,q$ are in $\Qbar[z]$ and $g$ is a purely transcendental
$E$-function, then for any $u\in \Qbar[z]$, the identity $f=p-qu+q(g+u)$ is
another admissible decomposition because $g+u$ is still a purely
transcendental $E$-function. In particular, the Taylor coefficients of
$f$ and
$g$ may lie in two different number fields. However, we have the following:

\begin{proposition}\label{prop:unicitydecomp} Every transcendental $E$-function (with coefficients in $\Qbar$) can be written in a unique way as $f=p+qg$ with $p,q\in\Qbar[z]$, $q$ monic and $q(0)\neq 0$, $\deg(p)<\deg(q)$ and $g$ a purely transcendental $E$-function.
\end{proposition}
We shall call this decomposition the {\em canonical decomposition of $f$}.
\begin{proof} 
As said above, it is proved in~\cite{rivoalnote2016} that any transcendental $E$-function $f$ can be written $P+QF$ where $P,Q\in\Qbar[z]$ and $F$ is a purely transcendental $E$-function. Note that since $f$ is transcendental, $Q$ is not identically zero: we can then write $Q(z)=z^m\sum_{k=0}^d q_k z^k$ with $q_0q_d\neq 0$ and $d,m\in \mathbb N$. The function $G(z)\coloneq q_dz^mF(z)$ is still a purely transcendental $E$-function and $QF=qG$ where $q(z)\coloneq \sum_{k=0}^d (q_k/q_d) z^k\in \Qbar[z]$  is monic and such that $q(0)\neq 0$. We now perform the Euclidean division of $P$ by $q$: we have $P=rq+p$ for some $p,r\in \Qbar[z]$ with $\deg(p)<\deg(q)$. Defining $g\coloneq G+r$, which is again a purely transcendental $E$-function, we observe that the decomposition $f=p+qg$ is of the form in the proposition.

\medskip

We now prove uniqueness of such a decomposition.
Consider two
decompositions $p+qg=\widetilde{p}+\widetilde{q}
\widetilde{g}$ of an $E$-function $f$ with
polynomials~$p,q,\widetilde{p},\widetilde{q}$ and $g,\widetilde{g}$
purely transcendental functions as in the statement of the Proposition.

We first prove that~$q=\widetilde{q}$. 
Obviously, these polynomials share the same
set of roots,
namely $\textrm{Exc}(f)$. Moreover $g$ and $\widetilde{g}$ being
purely transcendental, we claim that any root of $q$ and $\widetilde{q}$ has the same multiplicity in $q$ and in $\widetilde{q}$, so that $q$ and $\widetilde{q}$ are equal up to a non-zero constant factor, hence equal because they are both monic. To prove the claim, let $\rho$ be a root of $q$ of multiplicity $m$ and of multiplicity $\widetilde{m}$ for $\widetilde{q}$: if $m<\widetilde{m}$ then differentiating $m$ times both sides of $p+qg=\widetilde{p}+\widetilde{q}\widetilde{g}$  and evaluating at $z=\rho$, we obtain that $g(\rho)\in \Qbar$ which is not possible because $\rho\neq 0$, hence by symmetry of the situation we have $m=\widetilde{m}$.

Finally, since in the decompositions $p+qg=\widetilde{p}+q\widetilde{g}$,
we have $\deg(p)<\deg(q)$ and $\deg(\widetilde{p})<\deg(q)$, this forces  $p=\widetilde{p}$ and $g=\widetilde{g}$. Indeed,  $g-\widetilde{g}=(\widetilde{p}-p)/q\coloneq u$ is a rational $E$-function, i.e., a
polynomial in $\Qbar[z]$. Hence, consideration of the degree on both sides of $\widetilde{p}-p=uq$  forces $u=0$.
\end{proof}

Our main contribution in this section is to prove that when $f$ has
coefficients in $\mathbb
Q$, then we can find polynomials and a purely transcendental $E$-function
involved in the canonical decomposition that also have coefficients in~$\mathbb Q$, and
moreover that one can compute this decomposition algorithmically. More
precisely, we prove:

\begin{theorem} \label{theo:decomp} 
Any transcendental $E$-function $f\in\Q[[z]]$ admits a canonical decomposition $f=p+qg$, where
$p$ and $q$ are polynomials in $\Q[z]$ and $g\in\Q[[z]]$ is a purely
transcendental $E$-function. 
Moreover, if $f$ is given by a linear differential equation together with
sufficiently many initial terms, then one can effectively determine $p$ and
$q$.
\end{theorem}
Of course, once $p$ and $q$ are determined,  $g$ is determined by a linear differential equation together with sufficiently many initial terms, simply  because $g=(f-p)/q$. 
Before proceeding to the proof of~\cref{theo:decomp}, we state a very useful
fact.

\begin{proposition} \label{prop:decomp1}
Let $f\in\Q[[z]]$ be an $E$-function and let $\alpha \in \mathrm{Exc}(f)$. Then,
\begin{itemize}
	\item[$(i)$]  $f(\alpha)\in \mathbb Q(\alpha)$; 
	\item[$(ii)$]  all Galoisian conjugates of $\alpha$ belong to  $\mathrm{Exc}(f)$;
    \item[$(iii)$] for any Galoisian conjugate $\alpha'$ of $\alpha$,  
	the value $f(\alpha')$ is a Galoisian conjugate of~$f(\alpha)$. 
\end{itemize}	 
\end{proposition}

\begin{proof} The three statements are consequences of
Algorithm~\ref{algo:alg-values}. For any given root $\alpha$ of any
irreducible factor $\mu\in \mathbb Q[z]$ of $v_0$, the algorithm
determines if there exists a vector in~$\Qbar^{s+1}$ of the form $(\beta, -1,
0,\ldots, 0)$ which is in the left kernel of the matrix $\mathcal{M}(\alpha)$,
whose entries are in $\mathbb Q(\alpha)$. The existence of this vector is
equivalent to $f(\alpha)=\beta\in\Qbar$. When it exists, this proves that
$f(\alpha)=\beta\in \mathbb Q(\alpha)$, a fact proved in
\cite{FiRiCommentarii, FiRiEcoleX} in the general case (with $\mathbb Q(\alpha)$ replaced by $\mathbb K(\alpha)$ when $f$ has coefficients in a number field $\mathbb K$).
Now, such a vector exists if and only there exists a vector of the same form
in the left kernel of the matrix $\mathcal{M}(\alpha')$, where $\alpha'$ is
any Galoisian conjugate of $\alpha$ (i.e. any other root of $\mu$ in this
case). It follows that for any conjugate $\alpha'$ of $\alpha$, $f(\alpha')$
is a conjugate of $f(\alpha)$.
\end{proof}   

The existence of a decomposition as in \cref{theo:decomp} can be deduced
from \cref{prop:unicitydecomp}, by letting $\textup{Gal}(\Qbar/\mathbb Q)$
act on the decomposition delivered by \cref{prop:unicitydecomp} and by using
its uniqueness.
In order to obtain the polynomials~$p$ and~$q$ effectively, we propose a
``rational version'' of the proof in~\cite{rivoalnote2016}, which avoids
working in algebraic extensions. 

\begin{proof}[Proof of~\cref{theo:decomp}]   
The set $\textrm{Exc}(f)$  can be computed using \cref{algo:alg-values}. If $\textrm{Exc}(f)=\emptyset$, then the canonical decomposition of $f$ is $f=p+qg$ with $p=0$, $q=1$ and $g=f$. From now on, we assume that $\textrm{Exc}(f)= \{ \alpha_1, \ldots, \alpha_k \}\neq \emptyset$; 
by~\cref{prop:decomp1},  $\textrm{Exc}(f)$ can be partitioned into blocks of 
Galois conjugated values.

Let us first assume that there is only one such block, i.e. that  $\{\alpha_1, \ldots, \alpha_k\}$ is the set of roots of a monic irreducible polynomial $E \in \Q[z]$. 
For any $m\geq 0$, we consider the $E$-adic expansion of~$f$ to order~$m$:
\begin{equation} \label{eq:Eadic}
  f = p_0 + p_1 E + \cdots + p_m E^m + E^{m+1} g_m,
\end{equation}	                                                  
with $p_0(z), \ldots, p_m(z)\in \mathbb{C}[z]$ each of degree less than 
$k=\deg(E)$, and $g_m\in\mathbb{C}[[z]]$.

\medskip We will prove the following claims:

\smallskip  {\bf Claim 1.} $p_0, \ldots, p_m\in \Q[z]$
and $g_m\in\Q[[z]]$.

\smallskip  {\bf Claim 2.}  $g_m$ is an $E$-function.

\smallskip {\bf Claim 3.}  There exists an $m\geq 0$ such that $g_m$ is purely transcendental.

\smallskip From these claims, the proof of the first part of the theorem
follows by taking $p \coloneq  p_0 + p_1 E + \cdots + p_m E^m$, $q\coloneq 
E^{m+1}$ and $g\coloneq g_m$.

\medskip

 {\bf Proof of Claim 1.} It is enough to prove it for $m=0$, and then iterate.
We have $f = p_0 + E g_0$ with $p_0\in \mathbb{C}[z]$ of degree
less than $k$ and $g_0\in\mathbb{C}[[z]]$ and we need to prove that the
coefficients of $p_0$ and $g_0$ are actually in~$\Q$. First, we observe that
$p_0(z) = c_0 + c_1 z + \cdots + c_{k-1} z^{k-1}$ is the unique polynomial in
$\mathbb{C}[z]$ such that $p_0(\alpha_i) = f(\alpha_i)$ for $1\leq i \leq k$.
In matrix terms, this rewrites as
\[
\begin{pmatrix}
 1 & \alpha_1 & \ldots & \alpha_1^{k-1} \\
 \vdots & & & \vdots \\
 1 & \alpha_k & \ldots & \alpha_k^{k-1} \\
\end{pmatrix} \cdot 
\begin{pmatrix}
 c_0 \\
 \vdots\\
 c_{k-1}  \\
\end{pmatrix} 
=
\begin{pmatrix}
 f(\alpha_1) \\
 \vdots\\
 f(\alpha_{k-1})  \\
\end{pmatrix}
\]  
and by multiplying this equality on the left by the transpose of the Vandermonde matrix, we get the equivalent identity
\begin{equation}\label{eq:powersums}
\begin{pmatrix}
 k & \sum_i \alpha_i & \ldots & \sum_i\alpha_i^{k-1} \\
 \sum_i \alpha_i & \sum_i \alpha_i^2 & \ldots & \sum_i\alpha_i^{k} \\
 \vdots & & & \vdots \\
 \sum_i\alpha_i^{k-1} & \sum_i\alpha_i^{k} & \ldots & \sum_i\alpha_i^{2k-2} \\
\end{pmatrix} \cdot 
\begin{pmatrix}
 c_0 \\
 c_1 \\
 \vdots\\
 c_{k-1}  \\
\end{pmatrix} 
=
\begin{pmatrix}
 \sum_i f(\alpha_i) \\
 \sum_i \alpha_i f(\alpha_i) \\
 \vdots\\
 \sum_i \alpha_i^{k-1} f(\alpha_i) \\
\end{pmatrix}.
\end{equation}            
Now, the matrix on left-hand side of~\cref{eq:powersums} is invertible and with coefficients in~$\Q$, since it contains the power sums of the roots of the polynomial $E\in\Q[z]$.
On the other hand, for any $E$-function $g\in\Q[[z]]$, we have that $\sum_i g(\alpha_i) \in \Q$, by~\cref{prop:decomp1}. Applying this to the $E$-functions $f(z)$, $z f(z), \ldots, z^{k-1} f(z)$, we deduce that   the right-hand side of~\cref{eq:powersums} is a vector in~$\Q^k$. This implies that the $c_i$'s are all in~$\Q$, hence $p_0\in\Q[z]$. From there it directly follows that $g_0\in \Q[[z]]$.

\medskip {\bf Proof of Claim 2.}   It is again enough to prove the claim for $m=0$, and then iterate. 
Indeed, from~\cref{eq:Eadic} it follows that the $E$-adic expansion of
$g_{m-1}$ to order~1 is $g_{m-1} = p_m + E g_m$, and since $f = p_0 + E
g_0$ is an $E$-function one deduces iteratively that $g_0, g_1, \ldots, g_m$
are $E$-functions.

It remains to prove that if we have $f = p + E g$ with 
$E\in \Qbar[z]$ and $p\in\Qbar[z]$ of 
degree less than~$k=\deg(E)$ 
and $g\in\Qbar[[z]]$,
then $g$ is an $E$-function. This is done by induction on~$k\ge 1$.  For $k=1$,
this is precisely Property~{\bf (P2)}. Assume the property is proved for any $E$ of degree $k-1\ge 1$ and any $p$ of degree less than $k-1$. Assume we have $f = p + E g$ with 
$E\in \Qbar[z]$ of degree $k$, $p\in
\Qbar[z]$ of degree less than~$k$ and $g\in\Qbar[[z]]$. 
Let $\beta$ be one of the roots of $E$ and write $p(z)=\sum_{j=0}^{k-1} p_j (z-\beta)^j$. Then $f(\beta)=p_0$ and by Property~{\bf (P2)},  
\[
\frac{f(z)-f(\beta)}{z-\beta} = \sum_{j=0}^{k-2} p_{j+1} (z-\beta)^j + \frac{E(z)}{z-\beta} g(z)
\]
is an $E$-function. Since $E(z)/(z-\beta)$ is a polynomial of degree $k-1$ and $\sum_{j=0}^{k-2} p_{j+1} (z-\beta)^j$ is of degree less that $k-1$, we deduce that  $g$ is an $E$-function by the induction hypothesis.

\medskip {\bf Proof of Claim 3.}  From~\cref{eq:Eadic} it follows that the only exceptional values of the $g_m$'s are necessarily contained in the set $\textrm{Exc}(f) = \{\alpha_1, \ldots, \alpha_k\}$.     

We will show that there exists an~$m$ such that $g_m$ does not have any of the $\alpha_j$'s as an exceptional value, and therefore $g_m$ is purely transcendental. 
By~\cref{prop:decomp1}, this is equivalent to proving that the $g_m$'s cannot all share $\alpha\coloneq \alpha_1$ as an exceptional value.

Setting $g_{-1}=f$, it follows from~\cref{eq:Eadic} (with $m$ replaced by $m-1$, and then by differentiating $m$ times) that $1, f^{(m)}(\alpha)$
and $g_{m-1}(\alpha)$ are linearly dependent over~$\Qbar$ for all~$m\geq 0$ by
a relation of the form $f^{(m)}(\alpha)=u_m+v_m g_{m-1}(\alpha)$ with
$u_m,v_m\in\Qbar$ and $v_m\neq 0$. Hence, 
\[
\text{trdeg}_{\Qbar}\big(f(\alpha), \ldots, f^{(m)}(\alpha)\big)
=\text{trdeg}_{\Qbar}\big(g_{-1}(\alpha), \ldots, g_{m-1}(\alpha)\big)
\quad \text{for all} \; m\geq 0.
\] 
By contradiction, let us now assume that $g_m(\alpha) \in \Qbar$ for all~$m\geq -1$. Then we have 
$$
\text{trdeg}_{\Qbar}(f(\alpha), \ldots, f^{(m)}(\alpha))
=0
$$ for all $m\geq 0$. 
Now by Property {\bf (P1)} it follows that $f$ satisfies an LDE,
of some order~$\mu\geq 1$, having only $z=0$ as finite singularity. By considering 
the corresponding companion system $Y'= A Y$ where $f$ is the first element 
of the column vector $Y$, the matrix $A$ has Laurent polynomial entries 
in~$z$, hence the Siegel-Shidlovskii theorem ensures that 
\[
0 = \text{trdeg}_{\Qbar}\big(f(\alpha), \ldots, f^{(\mu-1)}(\alpha)\big)
=
\text{trdeg}_{\Qbar(z)}\big(f(z), \ldots, f^{(\mu-1)}(z)\big) \geq 1,
\]                                                           
a contradiction.    

On the effective side, note that one can compute the $E$-adic
expansion~\eqref{eq:Eadic} of $f$ to any order~$m$, for instance using
linear algebra. Then, to compute the needed decomposition, one may, for
increasing values $m=0,1,\ldots$, compute a linear differential equation for
$g_m$ as in~\eqref{eq:Eadic} together with sufficiently many initial terms,
and test using~\cref{algo:alg-values} whether $\mathrm{Exc}(g_m)$ is empty or
not. This procedure will eventually terminate.

We now treat the general case, where $\mathrm{Exc}(f)$ contains several blocks
$B_1,\ldots, B_p$, each block containing conjugated exceptional values. Denote
by $E_j(z)$ the minimal polynomial $\prod_{\alpha \in B_j} (z-\alpha)$ of the
elements in~$B_j$. By the reasoning used in the case of a single block, one
first finds a decomposition $f = p_1 + q_1 g_1$ with $p_1, q_1$ in $\Q[z]$ and
$g_1 \in\Q[[z]]$ an $E$-function such that $\mathrm{Exc}(g_1) =
\textrm{Exc}(f) \setminus B_1$.
Then, one applies the same to the $E$-function $g_1$, and writes it as $g_1 =
p_2 + q_2 g_2$, and thus $f = (p_1 + q_1 p_2) + (q_1 q_2) g_2$, with $p_2,
q_2$ in $\Q[z]$ and $g_2 \in\Q[[z]]$ an $E$-function such that
$\mathrm{Exc}(g_2) = \textrm{Exc}(f) \setminus (B_1 \cup B_2)$.
Continuing the same way~$p$ times, we end up with a decomposition $f = p + q
g$, with $p, q$ in $\Q[z]$ and $g \in\Q[[z]]$ an $E$-function such that
$\mathrm{Exc}(g) = \textrm{Exc}(f) \setminus (B_1 \cup \cdots \cup B_p) =
\emptyset$.
Moreover, by construction we have that 
$q$ monic, 
$q(0)\neq 0$ and $\deg(p)<\deg(q)$.
This concludes the proof.
\end{proof}

\subsection{\texorpdfstring{$E$}{E}-functions with coefficients in a  number field} \label{ssec:extensions1}
In general, an $E$-function is a power series
\[f(z)\coloneq \sum_{n=0}^{\infty} \frac{a_n}{n!} z^n\quad{\text{in}}\quad\Qbar[[z]]\]
such that
\begin{enumerate}
\item[$(i)$] $f(z)$ satisfies a homogeneous linear differential equation with 
coefficients in~$\Qbar(z)$;
\end{enumerate}
there exists $C>0$ such that
\begin{enumerate}
\item[$(ii)$] for any $\sigma \in \textup{Gal}(\Qbar/\mathbb Q)$ and any $n\ge 0$, $\vert \sigma(a_n)\vert \le C^{n+1}$;
\item[$(iii)$] for any $n\ge 0$, there exists $d_n \in \mathbb N\setminus \{ 0\} $ such that $d_n  \leq C^{n+1}$ and 
$d_na_m\in \mathcal{O}_{\Qbar}$ for all~$0\le m\le n$.
\end{enumerate}
In particular $(ii)$ with $\sigma=\textup{id}$ implies that $f(z)$ is an entire function. Moreover, $(i)$ implies that the coefficients $a_n$ all live in a certain number field, so that there are only finitely many Galoisian conjugates to consider in $(ii)$; if $a_n\in \mathbb Q$, this definition reduces to that of \cref{ssec:AR}.

The Adamczewski-Rivoal algorithm applies to these more general situations. The version stated in \cref{ssec:AR} also applies. Indeed, all the tools it uses work more generally. This is obviously the case for Beukers' desingularization, it is also the case for the algorithms used by minimization: greatest common right divisors, Hermite-Pad\'e approximants, series solutions and the computation of bounds on the degrees of factors (see \cite{BoRiSa21}).

\subsection{Siegel's original definition}  \label{ssec:extensions2}
$E$-functions with algebraic coefficients have been first defined by Siegel~\cite{siegel} in 1929 in a more general way: in $(ii)$ and $(iii)$ above, the upper bounds $(\cdots)\le C^{n+1}$ for all $n\ge 0$ are replaced by: for all $\varepsilon>0$, there exists $N(\varepsilon)$ such that $(\cdots)\le n!^{\varepsilon}$ for all $n\ge N(\varepsilon)$. $E$-functions considered above are sometimes denoted $E^*$-functions or called ``$E$-functions in the strict sense'': since Andr\'e's work~\cite{Andre00, Andre2}, it has become standard (though improper) to call them simply ``$E$-functions'' as well. 

The Siegel-Shidlovskii and Nesterenko-Shidlovskii theorems both hold in that setting.
The latter was refined by Beukers for $E$-functions in the strict sense only. Then, Andr\'e generalized Beukers' lifting theorem to $E$-functions in Siegel's sense by a completely different method \cite{Andre3}; another proof was later given by Lepetit \cite{lepetit} by a (non-trivial) adaptation of Beukers' original method. 

 We note here that Lepetit~\cite{lepetit} also generalized the Adamczewski-Rivoal algorithm to the case of $E$-functions in Siegel's original sense: he showed that all the steps in this algorithm work exactly the same {\em mutatis mutandis}, so that in fact our more efficient algorithm described here applies as well if the input is an $E$-function in Siegel's sense with rational coefficients. Moreover, the decomposition $f=p+qg$ studied in \cref{ssec:decompositionfpqg},  holds in Siegel's sense, in particular Theorem~\ref{theo:decomp}. However, it is conjectured that the classes of $E$-functions  in Siegel's sense and of $E$-functions in the strict sense are the same (see \cite[p.~715]{Andre00}). This implies that the distinction is in practice illusory because all known examples of $E$-functions satisfy all the conditions to be $E$-functions in the strict sense.

\section{Examples} \label{sec:examples}

\subsection{The Lorch-Muldoon example} \label{ssec:lorchmuldoon}
In a special case of a result due to Lorch and Muldoon~\cite{LoMu95}, the starting point is the following equation satisfied by the fourth derivative of  Bessel's~$J_0$ function:
\[
z(z^2-3)^2y''(z)
+ (z^2-15)(z^2-3)y'(z) 
+ z(z^4 - 10z^2 + 45)y(z) 
= 0,
\]
with initial conditions $y(0)=3/8,y'(0)=0$.
With this input, \cref{algo:alg-values} returns
\[y(\pm\sqrt3)=0,\]
showing that the only non-zero algebraic points where the $E$-function ~$J_0^{(4)}$ is algebraic are~$\pm\sqrt{3}$, where it vanishes.
Moreover, the algorithm described in~\cref{ssec:decompositionfpqg} provides the canonical decomposition $J_0^{(4)} = p + q g$, where
\begin{equation}\label{LMdecomp}
p(z) = 0, \; 
q(z) = z^2-3\quad
\text{and} \quad
g(z) = -J_2(z)/z^2.
\end{equation}
Here, the purely transcendental\footnote{Siegel first proved that Bessel's function $J_2$ is purely transcendental in \cite [pp.~31-32, \S 4]{siegel}.} 
$E$-function $g$ is given by the differential equation
\[
y''(z)
+ 5 y'(z) 
+ zy(z) 
= 0,\]
with initial conditions $y(0)=-1/8,y'(0)=0$. The decomposition~\eqref{LMdecomp} explains the {\em a priori} unexpected fact that $\mathrm{Exc}(J_0^{(4)}) = \{ \pm \sqrt{3} \}$.

\medskip It is easy to construct $E$-functions that take algebraic values at certain chosen  algebraic points: consider $p+qg$ where $p, q\in \Q[z]$ and $g$ is any $E$-function in $\Q[[z]]$. Conversely, as we have seen in~\cref{theo:decomp},  any $E$-function $f\in\Q[[z]]$ can be written $f=p+qg$ where $p, q\in \Q[z]$ and $g$ is a purely transcendental $E$-function.

It turns out to be difficult to find an $E$-function which takes an algebraic value at a non-zero algebraic point and which is not obviously of the form $p+qg$ as above. The goal of the next two sections is to provide two infinite families of $E$-functions for which we believe it is difficult to guess {\em a priori} \cref{prop:newfam1F1,prop:lomu} below. Their proof is inspired in part by that of the evaluation $J_0^{(4)}(\pm \sqrt{3})=0$ above. 
Besides their theoretical interest, we used these propositions to check the correctness of various routines of our algorithms.

\subsection{A first family of \texorpdfstring{$E$}{E}-functions}\label{sec:newfam1}

We start with a result on the exceptional values of an infinite family of ${}_1F_1$ functions.

\begin{proposition} \label{prop:newfam1F1}
Let $a\in \mathbb Q\setminus \mathbb Z_{\le 0}$
and $d\in\mathbb{N}$.
Then:
\begin{enumerate}
\item[$(i)$] $R(z)\coloneq \sum_{k=0}^d \binom{d}{k}(a)_k z^{d-k}$ has $d$ simple roots;
\item[$(ii)$]  $\mathrm{Exc}({}_1F_1[d+1;a+d+1; -z])$ coincides with the set of roots of $R$;
\item[$(iii)$] for any root $\rho$ of $R$, the following identity holds:
\begin{equation}\label{eq:remark:s}
{}_1F_1[d+1;a+d+1; -\rho]
=
-\frac{(a)_{d+1}}{\rho R'(\rho)} \cdot
\end{equation}
\end{enumerate}
\end{proposition}

\begin{proof}
Let us first treat the case $d=0$. Then $(i)$ and $(iii)$ trivially hold since $R = 1$.
On the other hand, by \cite[p.~185, Theorem 1]{shidl}, we have that 
	$\mathrm{Exc}({}_1F_1[1;a+1; -z])=\emptyset$ for all $a\in \mathbb Q\setminus \mathbb Z_{\le 0}$, and this proves $(ii)$. 
	
We assume in the rest of the proof that $d\in\mathbb{N}\setminus\{ 0 \}$.
From the differential equation 
$zy''(z) + (a+z+d+1)y'(z)+(d+1)y(z)=0$
satisfied by 
${}_1F_1[d+1; a+d+1; -z ]$,
\cref{algo:min-inhom} computes its adjoint  
$zy''(z) - (z+a+d-1)y'(z)+dy(z)=0$
and discovers that it admits $R$ as a  non-zero polynomial solution.
From there, it computes
$b_1 = z R(z)$ and $b_0 = (a+z+d+1) R(z) - b_1'(z)
$, with the property that 
${}_1F_1[d+1; a+d+1; -z ]$
is a solution of the inhomogeneous differential equation 
\begin{equation}\label{eq:b1b0a}
b_1(z) y'(z) + b_0(z) y(z) = a(a+1)\cdots (a+d).
\end{equation}
It follows from~\eqref{eq:b1b0a} that $R$ has only simple roots,  since if $R(\rho)=R'(\rho)=0$, then $b_1(\rho)=b_0(\rho)=0$, hence $a\in \{0,-1,\ldots, -d\}$, which is impossible. This proves $(i)$.

Next, \cref{algo:alg-values} evaluates~\eqref{eq:b1b0a} at the roots $\rho$ of $R$.
Since $b_1(\rho)=0$, it follows that 
$
{}_1F_1[d+1;a+d+1; -\rho]
= a(a+1)\cdots (a+d)/b_0(\rho)$,
and hence $(iii)$ holds.

To prove $(ii)$, note that $f(z)\coloneq {}_1F_1[d+1;a+d+1; -z]$ is a transcendental $E$-function such that $1,f,f'$ are linearly dependent over $\mathbb Q(z)$ (as Eq.~\eqref{eq:b1b0a} shows), and $(f,f')^{\mathrm T}$ is solution of a differential system with only 0 as singularity. In particular, by the Siegel-Shidlovskii theorem,   for any $\alpha \in \Qbar^*$ such that $f(\alpha)\in \Qbar$ we have $f'(\alpha) \notin \Qbar$, and consequently the differential equation \eqref{eq:b1b0a} shows that $\mathrm{Exc}({}_1F_1[d+1;a+d+1; -z])$ coincides with the set of roots of $R$, proving $(ii)$.
\end{proof}

\medskip 
\begin{remark} \label{rem:1F1}
	\emph{Let us now make several remarks on \cref{prop:newfam1F1}.}
\end{remark} 

\begin{enumerate}[label=(\alph*),left=0pt]
\smallskip 
\item We do not know if the evaluation~\eqref{eq:remark:s} is available in the (very rich) literature on special functions. 
It is remarkable that it was discovered (and proved) using our algorithms. Note that Eq.~\eqref{eq:remark:s} holds more generally for all $a\in \mathbb C\setminus \mathbb Z_{\le 0} $.

\smallskip 
\item The polynomial 
$R(z)$ in \cref{prop:newfam1F1}
is equal to 
$(a)_d \cdot {}_1F_1[ -d; 1-a-d; z ]$, thus
it can be expressed in terms of generalized Laguerre polynomials as
\[R(z) = (-1)^d d! \cdot L_{d}^{(-a-d)}(z) .\]
As proved by Schur~\cite{Schur1931}, the discriminant of $R$ is equal to
$\prod_{j=2}^d j^j (j-a-d)^{j-1}$, which is non-zero since $a\notin \Z_{\le 0}$; this yields a different proof that $R$ has only single roots.

\smallskip 
\item If $a\in \mathbb Z_{\le 0}$, then the situation is simpler and well understood.
Indeed, formula
\href{https://functions.wolfram.com/HypergeometricFunctions/Hypergeometric1F1/03/01/02/0007/}{07.20.03.0007.01} on Wolfram's mathematical functions site implies
\[
{}_1F_1[d+1;a+d+1; -z]
 = 
\frac{\left(-1\right)^{-a} \left(-a \right)! }{\left(-d \right)_{-a}}
\cdot e^{-z} \cdot L_{-a}^{\left(d +a \right)}\! \left(z \right)
\]
and in particular
\[
{}_1F_1[d+1; d; -z] = e^{-z} \cdot \frac{d-z}{d} \quad \text{for any} \; d \neq 0
\]
and
\[
{}_1F_1[d+1; d-1; -z] = 
e^{-z} \cdot \frac{ (d - z)^2 - d}{d \left(d -1\right)}
\quad \text{for any} \; d \notin \{ 0, 1 \}.
\]
In turn, these functional identities induce \emph{infinite} families of numerical identities, such as 
\begin{equation} \label{id-simple-1}
{}_1F_1[d+1; d; -d] = 0 \quad \text{for any} \; d \in \mathbb{N} \setminus \{ 0 \}
\end{equation}
and
\begin{equation} \label{id-simple-2}
{}_1F_1[d^2+1; d^2-1; d-d^2] = 
{}_1F_1[d^2+1; d^2-1; -d-d^2] =
0 \quad \text{for any} \; d \in \mathbb{N} \setminus \{ 0, 1 \}.
\end{equation}

More generally,
$\mathrm{Exc}({}_1F_1[d+1;a+d+1; -z])$
coincides with the set of roots of $L_{-a}^{\left(d +a \right)}\! \left(z \right)$.

\smallskip 
\item  When $d=1$, the rational canonical decomposition of $f(z)\coloneq {}_1F_1[2;a+2; -z]$ given by \cref{theo:decomp}  is
$f=p+qg$ with $p=a+1$, $q=z+a$ and $g = -{}_1F_1[1;a+2; -z]$ (note that 
$g$ is purely transcendental by remark (2) above).

Decompositions of $f(z)\coloneq {}_1F_1[d+1;a+d+1; -z]$ can easily be written down when $d\ge 2$ but they are neither as explicit nor necessarily canonical. Since $b_0$ and $b_1$ (in the proof of \cref{prop:newfam1F1}) are coprime, there exist $u,v\in \mathbb Q[z]$ such that $b_1u+b_0v=1$. Then we have the decomposition  $f=(a)_{d+1}v+Rg$, where $g(z)\coloneq z(u(z)f(z)-v(z)f'(z))$ is purely transcendental. Indeed, the decomposition is immediate to check and let $\alpha \in \Qbar^*$ be such that $g(\alpha)\in \Qbar$. Then $f(\alpha)\in \Qbar$ as well, hence $b_1(\alpha)=0$ by $(ii)$ in Proposition \ref{prop:newfam1F1}, so that $v(\alpha)\neq 0$ by the relation  $b_1u+b_0v=1$. Therefore $g(\alpha)\notin \Qbar$ because $f'(\alpha)\notin \Qbar$. This contradiction proves that there is no such $\alpha$.

When $d=1$, this procedure provides an alternative way to obtain the above canonical decomposition of $f(z)\coloneq {}_1F_1[2;a+2; -z]$,
with $g$ represented as $g(z)=z/(a(a+1)) \cdot ((z+a-1) f(z)+(z-1) f'(z))-1$.

When $d=2$, it provides a decomposition of $f(z)\coloneq {}_1F_1[3;a+3; -z]$ as $f = p+qg$, where $p(z)=z^2/2 + (a-2)z/2+1$, $q(z)=z^2+2az+a(a+1)$ and $g(z)=
-z/(2a(a+1)(a+2)) \cdot ((z^2+2(a-1)z+a^2+2)f+(z^2+(a-2)z+2)f')$. 
The canonical decomposition of $f$ is then readily obtained as $f = \tilde{p}+q \tilde{g}$, where
$\tilde{p}\coloneq p-q/2$ and $\tilde{g}\coloneq g+1/2$.

\smallskip 
\item  When $d=2$, \cref{prop:newfam1F1} implies the following evaluation:
\begin{equation}\label{eq:remark:s=2}
e^{a\mp i\sqrt{a}} {}_1F_1[a;a+3; -a \pm i\sqrt{a} ] =
{}_1F_1[3; a+3; a \mp i\sqrt{a} ] =
(a+2)(1\mp i\sqrt{a})/2.
\end{equation}
The left-hand side  is a special case of Kummer's identity 
$e^{-z} {}_1F_1[a;b;z]={}_1F_1[b-a;b;-z]$.
The right-hand side follows from the fact that the roots of $R(z) = z^2+2az+a(a+1)$ are $\{ -a \pm i \sqrt{a}\}$, since then  \cref{prop:newfam1F1} implies
$
 {}_1F_1[3; a+3; \rho ] = 
-a(a+1)(a+2)/(\rho R'(\rho))
=(a+1)(a+2)/(2(\rho+a+1))$.

\smallskip 
\item  Generalized Laguerre polynomials are most of the time irreducible in $\Q[z]$, but not always. Filaseta and Lam~\cite[Thm.~1]{FiLa02} proved that if $\alpha\notin \Z_{< 0}$,
then $L_d^{(\alpha)}(z)$ is irreducible in $\Q[z]$ for sufficiently large~$d$. 
However, for some values $d,\alpha$, the polynomial $L_d^{(\alpha)}(z)$ can be reducible.
This is so  e.g. for 
$d=5, \alpha=7/5$, 
for which $L_d^{(\alpha)}(z)$ admits the linear factor 
$z-12/5$. 
This observation leads to simple particular cases of~\eqref{eq:remark:s} such as:
\begin{align} 
\label{eq:eval_rat_3_46}
{}_1F_1[4;703/725; -312/725] &= -20999/525625, \\
 \label{eq:eval_rat_4_1}
{}_1F_1[5;-113/3; -140/3] &= -30073/27,\\
  \label{eq:eval_rat_5_1}
{}_1F_1[6;-2/5; -12/5] &= 1309/625,\\
 \label{eq:eval_rat_5_2}
{}_1F_1[6; 314/63; 20/63]& = 365707/250047,\\
 \label{eq:eval_rat_7_1}
{}_1F_1[8; 48/7; 6/7]& = 45305/16807.
\end{align}
Classifying all pairs $(d,a) \in \mathbb{N} \times \Q$ such that the $E$-function ${}_1F_1[d+1;a+d+1; -z] $ takes algebraic values at \emph{rational} points~$z$ is a non-trivial task, since by \cref{prop:newfam1F1} this is equivalent to finding  $(d,a) \in \mathbb{N} \times \Q$ such that $L_{d}^{(-a-d)}(z)$ admits a rational root.

For $2 \leq d \leq 10$,
we  systematically searched for such ``rational evaluations'' arising from 
\emph{reducible} Laguerre polynomials $L_d^{(-a-d)}$.
In particular, for all $2 \leq d \leq 10$, we looked for $z_0\in\Q$ 
with numerator and denominator between $-1000$ and $1000$, and 
such that $L_d^{(-a-d)}(z_0) \in \Q[a]$ has a root $a_0\in\Q$.
Each pair $(a_0,z_0)$ then yields a rational ${}_1F_1$-evaluation as above.

With $d\in \{ 2, 3, 4 \}$, we could find many such identities, for instance~\eqref{eq:eval_rat_3_46} and \eqref{eq:eval_rat_4_1}. 
By contrast, for $d \geq 5$, these rational evaluations become quite rare.
Of course, there are still infinite families of \emph{trivial} identities such as 
${}_1F_1[d+1; a; 0] = 1$ or of \emph{simple} identities such as \eqref{id-simple-1} and~\eqref{id-simple-2}.
But
other rational evaluations are rare with $d\geq 5$.
For instance, with $d=5$ we only found seven non-trivial  identities, of which~\eqref{eq:eval_rat_5_1} is the simplest and \eqref{eq:eval_rat_5_2} the most complex.
For $6 \leq d \leq 10$, the only non-trivial rational evaluation we found is~\eqref{eq:eval_rat_7_1}.

For $d=2$, all rational identities that we have found belong to the following infinite family:
\begin{equation} \label{eq:eval_rat_2_family}
{}_1F_1
\left[ 3; \frac{11}{4}-q^2-q; -\frac{(2q+3)(2q+1)}{4} \right] = 
\frac{(2q-1)(4q^2+4q-7)}{16}, \quad q \in \mathbb{Q}.
\end{equation}
Note that this identity is a particular case of~\eqref{eq:remark:s=2}. 

For $d=3$, all rational identities that we have found belong to the following parametric identity, which
specializes 
to~\eqref{eq:eval_rat_3_46} for $q=13/5$:
\begin{equation} \label{eq:eval_rat_3_family}
{}_1F_1 \left[ 4; \frac{q^3-12q+8}{2-3q}; \frac{q(1-q)(2-q)}{2-3q} \right] = -\frac{(q+2)(q^2+2q-2)(q^3-9q+6)}{6(2-3q)^2}.
\end{equation}


A unified way to prove Eqs.~\eqref{eq:eval_rat_2_family}--\eqref{eq:eval_rat_3_family} 
is by using that for $d \in \{2, 3 \}$, the curve (in $a$ and~$z$) defined by the generalized Laguerre polynomial $L_d^{(a)}(z)$ has genus $0$, and by using the connection of the $L_d^{(a)}$'s to our $_1F_1$’s. E.g., for $d=3$ we have the parametrization 
\[ \big\{ a = (q^3-9q+6)/(3q-2), \; z = q(q-1)(q-2)/(3q-2) \big\} 
\quad \text{for} \quad L_d^{(a)}(z)=0
 \] 
from which the above evaluation \eqref{eq:eval_rat_3_family} follows.

\smallskip 
\item When $d=4$, there is a nice connection between rational evaluations and elliptic curves.
The polynomial $R(z)$ in \cref{prop:newfam1F1} is equal to
$R(z) =  4! \cdot L_{4}^{(-a-4)}(z).$
Thus, $R(-z-a) = z^4 + 6 a z^2 - 8 a z+3 a (a+2)$ defines an elliptic curve $(\mathcal{E})$ whose 
Weierstrass form is $(\mathcal{W})$ $z^2 = a^3 - 76/3 a + 3440/27$. 
The Mordell–Weil group of $(\mathcal{W})$ is isomorphic to $\mathbb{Z}/2\mathbb{Z} \times \mathbb{Z}$, with generators 
	\[ \left\langle P_0 = (-20/3, 0), P_1 = (34/3, 36) \right\rangle 
	=  \left\langle P_0, P_2 = (-2/3, 12) \right\rangle,
	\]
where $P_0$ is a torsion point of order 2, and $P_1 + P_2 = P_0$.	
Each rational point on the elliptic curve $(\mathcal{W})$ 
gives rise to a non-trivial evaluation such as~\eqref{eq:eval_rat_4_1}.	
The points $P_0$ and $P_1$ themselves yield the trivial evaluations
$_1F_1[5; 3; -2] = 0$ and 
$_1F_1[5; 2; 0] = 1$,
while the point $P_2$ yields the undefined evaluation $_1F_1[5; -4; -6].$ 
However, their multiples yield interesting evaluations.
For instance, the rational point 
$-4.P_1 = (-53/12, 99/8)$ on  $(\mathcal{W})$ 
yields the rational point 
$(a,z) = (-128/3, -4)$ on $(\mathcal{E})$, 
which in turn provides identity~\eqref{eq:eval_rat_4_1}.

Similarly, the points $2.P_1 = (7/3, 9), -P_1 = (34/3, -36), -3.P_1 =  (-14/3, -12)$ and
$2.P_2 =  (7/3, -9)$ on $(\mathcal{W})$ 
yield the points
$(-8/3, 2), (-1/3, 1), (-27/25, 3/5)$ and $(-24/25, 6/5)$
on~$(\mathcal{E})$,  which in turn yield the rational evaluations
\begin{align} \label{eq:eval_rat_4_more}
{}_1F_1[5; 7/3; -2/3] &= 5/27,\\\nonumber
{}_1F_1[5; 14/3; 2/3] &= 55/27,\\\nonumber
{}_1F_1[5; 98/25; -12/25] &= 1679/3125,\\\nonumber
{}_1F_1[5; 101/25; 6/25] &= 4199/3125.\nonumber
\end{align}
Mordell's theorem~\cite{Mordell1922} (see also~\cite[Part C]{HiSi00}) implies that there are infinitely many (non-trivial) rational evaluations of the form ${}_1F_1[5; \alpha ; \beta ] = \gamma$ with $\alpha, \beta, \gamma\in\Q$, such as the five evaluations in~\eqref{eq:eval_rat_4_1} and \eqref{eq:eval_rat_4_more}.

On the other hand, for any $d\geq 0$, the genus of the curve (in $a$ and $z$) defined by the generalized Laguerre polynomial $L_{d}^{(-a-d)}(z)$ is equal to  $\lfloor (d/2-1)^2\rfloor$~\cite{Wong05} and hence at least~$2$ for $d\geq 5$ (see also~\cite[Prop.~4]{HaWo06}). It follows from \cref{prop:newfam1F1}, from \cref{rem:1F1}(b) and from Faltings' theorem~\cite{Faltings83} (see also~\cite[\S E.1]{HiSi00}) that, for any $d\geq 5$, there are finitely many evaluations of the form ${}_1F_1[d+1;  \alpha ; \beta ] = \gamma$ with $\alpha, \beta, \gamma\in\Q$, such that $\beta \neq 0$ and $\alpha \in \mathbb Q\setminus \mathbb Z_{\le d+1}$.
Similarly, from \cref{rem:1F1}(c) it follows that there are finitely many evaluations of the form ${}_1F_1[d+1;  d-\alpha ; \beta ] = \gamma$ with $\beta, \gamma\in\Q$, such that $\beta \neq 0$ and $\alpha \in \mathbb Z_{\geq 4}$.

\smallskip  
\item   With other choices such as $d=4$, $\alpha=12/5$, the polynomial $L_d^{(\alpha)}$ has non-linear factors. In particular, we obtain the quadratic irrational evaluation
\begin{equation} \label{eq:eval_alg_4_2}
{}_1F_1[5;-7/5; (6 \sqrt{15}-42)/5]
=
{11}/{5}+{66 \sqrt{15}}/{125},
\end{equation}
which is not a particular case of~\eqref{eq:remark:s=2}.

A similarly looking identity, and perhaps even more striking, is the quartic evaluation
\begin{equation} \label{eq:eval_alg_5}
{}_1F_1 \left[6; \frac{23-\sqrt{725-20 \sqrt{985}}}{10}; -\frac{6}{5}\right]
=
\frac{111 \sqrt{985} - 3 \sqrt{3353450-106670 \sqrt{985}}-3533 }{2500},
\end{equation}
although in~\eqref{eq:eval_alg_5} the ${}_1F_1$ function on the left is not an $E$-function anymore.

Another quartic evaluation, this time involving an $E$-function again, is
\begin{equation} \label{eq:eval_alg_5_2}
{}_1F_1 \left[6; \frac{22}{5}; -\frac{2}{5}\alpha\right]
=
\frac{17}{6255}\beta,
\end{equation}
where $\alpha \approx 5.15$ satisfies
$\alpha^4-21 \alpha^3+81 \alpha^2-21 \alpha+126=0$ and
$\beta \approx 1.59$ satisfies
$\beta^4-62 \beta^3+1584 \beta^2-2338 \beta-49=0$.


We were unable to locate in previous works any of the identities
\eqref{eq:remark:s=2}--\eqref{eq:eval_alg_5_2}, including in online encyclopedias such as \href{https://functions.wolfram.com/HypergeometricFunctions/Hypergeometric1F1/03/}{Wolfram's mathematical functions site} and the
\href{https://dlmf.nist.gov/13}{Digital Library of Mathematical Functions}.
Given how vast the literature on special functions is, we would not be surprised that some of these identities were already tabulated.

\smallskip 
\item 
\Cref{eq:remark:s} can also be proved directly starting from the relation between these ${}_1F_1$ and the incomplete gamma function~\cite[13.6.5]{dlmf}, \cite[13.6.10]{abraste}:
\[f(z)\coloneq \frac1a{}_1F_1[1; a+1; -z]=(-z)^{-a}e^{-z}\gamma(a,-z).\]
Successive differentiation of the hypergeometric series shows that
\[f^{(d)}(z)=\frac{(-1)^dd!}{(a)_d}{}_1F_1[d+1; a+d+1; -z].\]
On the other hand, we have 
$\left(\gamma(a,-z)\right)'=(-z)^ae^z/z$ by~\cite[8.1]{dlmf}, \cite[6.5.2]{abraste}.
Thus, by induction, there are two families of polynomials $(R_d)$ and~$(Q_d)$ such that
\begin{equation}\label{eq:id-1f1}
\frac{(-1)^dd!}{(a)_d}{}_1F_1[d+1; a+d+1; -z]
=
R_d(z)e^{-z}(-z)^{-a-d}\gamma(a,-z)+(-z)^{-d}Q_d(z)
\end{equation}
with 
\[R_d(z)=e^z(-z)^{a+d}\left((-z)^{-a}e^{-z}\right)^{(d)},\quad Q_{d+1}=dQ_d-zQ_d'+R_d,\quad Q_0=0.\]
The polynomial~$R_d$ is exactly the Laguerre polynomial $R$ from before. 

Thus, \cref{eq:remark:s} boils down to an evaluation of this more general formula at a root~$\rho$ of~$R_d$, giving
\[{}_1F_1[d+1; a+d+1; -\rho] = \frac{(a)_dQ_d(\rho)}{d!\rho^d}.\]
This is not exactly the same formula as above. The proof is completed by proving by induction that both~$Q_d$ and~$R_d$ satisfy the same recurrence
$u_{d+1}=(z+d+a)u_d-zdu_{d-1}$,
giving an explicit evaluation of the determinant
\[\left|\begin{matrix}Q_{d+1}&R_{d+1}\\ Q_d&R_d\end{matrix}\right|=\left|\begin{matrix}z+d+a&-zd\\ 1&0\end{matrix}\right|\cdot \left|\begin{matrix}Q_{d}&R_{d}\\ Q_{d-1}&R_{d-1}\end{matrix}\right|=\dots=d!z^d.\]
Evaluating at~$z=\rho$ gives $Q_d(\rho)/(d!\rho^d)=-1/R_{d+1}(\rho)$. 
This gives another simple expression for the right-hand side of~\cref{eq:remark:s}, which follows from $R_{d+1}(z)=(z+d+a)R_d(z)-zR_d'(z)$.
\end{enumerate}

\subsection{A second family of \texorpdfstring{$E$}{E}-functions}\label{sec:newfam2}

The next result considers exceptional values of second derivatives of products of ${}_1F_1$ with the exponential function. We recall that $J_0(-iz/2)=e^{-z/2}{}_1F_1[1/2;1;z]$ is such a product.
\begin{proposition} \label{prop:lomu}
Let $c\in \Qbar^*$  and $a,b\in \mathbb Q\setminus \mathbb Z_{\le 0}$ with  $a-b\notin \mathbb N$. Let $F(z)\coloneq e^{-cz} {}_1F_1[a;b;z]$. 
Then $\mathrm{Exc}(F'')=\emptyset$, except in the following (disjoint) cases:
\begin{enumerate}
	\item[{\bf 1.}] if $b=a(2c-1)/c^2$, then $\mathrm{Exc}(F'')=\{-a/c^2\}$ and $F''(-a/c^2)=0$;
	\item[{\bf 2.}] if $c=1$ and $b=a+1$, then $\mathrm{Exc}(F'')=\{-a\pm i\sqrt{a} \}$ and $F''(-a\pm i\sqrt{a}) = 1/(1\pm i\sqrt{a})$.
\end{enumerate}	
\end{proposition}

\begin{proof}

With the assumptions on $a,b,c$, we prove below that:

\smallskip
{\bf Fact 1.} $F$ and $F'$ are linearly independent over $\Qbar(z)$, i.e. $F$ does not satisfy any homogeneous LDE of order less than~$2$.

\smallskip
{\bf Fact 2.} $F$ satisfies the second-order LDE
\begin{equation}\label{eq:ODE2}
zF''(z)=(z-2cz-b)F'(z)+(cz+a-c^2z-cb)F(z),
\end{equation}
and it does not satisfy any inhomogeneous LDE of order 1, unless $c=1$ and $b=a+1$.

\medskip 
Postponing for a moment the proof of these facts, we distinguish two cases.

\medskip 
{\bf Case 1.} We first assume that either $c\neq 1$ or $b\neq a+1$.
By {\bf Fact 1}, the function $F$ is a non-polynomial $E$-function (hence a  transcendental one).
By {\bf Fact 2} and by Beukers' Corollary 1.4 of \cite{Beukers06}, it follows that the numbers $1$, $F(\xi)$ and $F'(\xi)$ are linearly independent over $\Qbar$ for any $\xi\in \Qbar^*$.

Assume now that $b\neq (2ac-a)/c^2$.
Then, for any $\xi\in \mathbb C$, the numbers $\xi-2c\xi-b$ and $c\xi+a-\xi c^2-cb$ 
cannot be simultaneously equal to 0. 
Since the numbers $1$, $F(\xi)$ and $F'(\xi)$ are linearly independent over $\Qbar$ 
for any $\xi\in \Qbar^*$, it follows that 
\[
F''(\xi)=\frac{1}{\xi}\big((\xi-2c\xi-b)F'(\xi)+(c\xi+a-\xi c^2-cb)F(\xi)\big)\notin \Qbar.
\]
Hence $\mathrm{Exc}(F'')=\emptyset$ when $b\neq (2ac-a)/c^2$. 

It remains to treat the sub-case $b\coloneq (2ac-a)/c^2$. With $z\coloneq -a/c^2$, we see that 
\[
z-2cz-b = cz+a-c^2z-cb =0, 
\]
so that $F''(-a/c^2)=0$. 
Since $z-2cz-b$ and $cz+a-c^2z-cb$ vanish  simultaneously for no other value of $z$ 
and since the numbers $1$, $F(\xi)$ and $F'(\xi)$ are linearly independent over $\Qbar$ for any $\xi\in \Qbar\setminus \{0,-a/c^2\}$, we deduce that $F''(\xi)\notin \Qbar$ for such $\xi$. Hence $\mathrm{Exc}(F'')=\{-a/c^2\}$ when $b=(2ac-a)/c^2$. The proposition is thus proved in {\bf Case 1}.

\medskip 
{\bf Case 2.} We assume that $c=1$ and $b = a+1$.
Then, $F''(z) = \frac{2}{(a+1)(a+2)} e^{-z} {}_1F_1[a;a+3;z]$,
so $\mathrm{Exc}(F'') = \mathrm{Exc}(e^{-z} {}_1F_1[a;a+3;z])=\{-a\pm i\sqrt{a} \}$.
The last equality is a consequence of 
\cref{prop:newfam1F1} with $d=2$.
The equality
$F''(-a\pm i\sqrt{a}) = 1/(1\pm i\sqrt{a})$
follows from~\eqref{eq:remark:s=2}.

\medskip {\bf Proof of Fact 1.} When $a,b\in\mathbb Q\setminus \mathbb Z_{\le 0}$ and $a-b\notin \mathbb N$, the asymptotic behavior of ${}_1F_1[a;b;z]$ 
as $z\to \infty$ 
(with $-\pi/2<\arg(z)< 3\pi/2$)  
is given  in \cite[p.~508, Eq.  13.5.1]{abraste}, \cite[13.7.2]{dlmf}. In the particular cases $z\to \pm \infty$, it reads
\[
{}_1F_1[a;b;z] \sim_{z\to +\infty} \frac{\Gamma(b)}{\Gamma(a)}e^{z}z^{a-b}\quad\text{and}\quad
{}_1F_1[a;b;z] \sim_{z\to -\infty} \frac{\Gamma(b)}{\Gamma(b-a)}e^{i\pi a} z^a.
\]
This rules out the possibility that $e^{-cz}{}_1F_1[a;b;z]$ satisfies a differential equation of order 1 over $\Qbar(z)$. 

\smallskip Note that a different (purely algebraic) proof is possible, based on a reasoning similar to the one in the statement and proof of~\cite[Lemma 4.2]{BBH88}.

\smallskip  {\bf Proof of Fact 2.}
Since ${}_1F_1[a;b;z]$ satisfies
$zy''(z) + (b-z)y'(z)-ay(z)=0$ and $e^{-cz}$ satisfies $y'(z)+ cy(z)=0$,
it follows by a simple computation that $F$ satisfies~\eqref{eq:ODE2}.
By {\bf Fact 1}, \eqref{eq:ODE2} is the minimal-order homogeneous LDE satisfied by~$F$.

Assume now that $F$ satisfies an inhomogeneous LDE of order~$1$.
We will follow the reasoning in~\cref{algo:min-inhom}, and show that the adjoint of~\eqref{eq:ODE2} does not possess any non-zero rational solutions in $\Qbar(z)$, unless $c=1$ and $b=a+1$.

The adjoint equation of~\eqref{eq:ODE2} writes
\begin{equation}\label{eq:ODE2adj}
z y''(z) + ((1-2 c) z - b + 2) y'(z) + (c (c - 1) z + b c - a - 2 c + 1) y(z) = 0.
\end{equation}
If it admits a rational solution $R(z)\in\Qbar(z)$, then the only potential pole of $R$ can be located at $z=0$.
The indicial equation of~\eqref{eq:ODE2adj} at $z=0$ is $s (s - b + 1)$. Hence, the possible valuations at $z=0$ of $R$ are $0$ and $b-1$. 
Since $b\in \mathbb Q\setminus \mathbb Z_{\le 0}$, this implies that $R$ is actually a polynomial solution in $\Qbar[z]$ of~\eqref{eq:ODE2adj}.
If $c\neq 1$, then the indicial polynomial at infinity of~\eqref{eq:ODE2adj} is a non-zero constant, equal to $c^2-c$; therefore in that case, $R$ cannot be a polynomial solution. It follows that $c=1$. Now, the indicial polynomial at infinity of~\eqref{eq:ODE2adj} is $s - a + b - 1$, hence the only possible degree of $R$ is
$a - b + 1$. Since $a-b\notin \mathbb N$, this implies that $b=a+1$ and that $R$ is a constant in~$\Qbar$. 
In this case, \eqref{eq:ODE2adj} admits the rational solution $y(z)=1$,
and $F$ satisfies $z F'(z)+(z+a) F(z) = a$.
\end{proof}

Note that, in the spirit of Proposition \ref{prop:lomu}, the following
examples that can be treated along the same lines: the third derivatives of
\[
 e^{-z/9} {}_2F_2[1/144,1/144;-7/16,-7/16;z] \quad \textup{and} \quad e^{-z/3} {}_2F_2[1/4,3/4;5/4,-9/4;z]
\]
vanish at $z=3/16$ and $z=-9/4$, respectively. Our minimization algorithm
finds their minimal differential equations, which are too big to be written
here: their (order, degree) are $(3,8)$ and $(2,7)$, respectively.
\cref{algo:alg-values} then shows that they are transcendental functions and
that $\{ 3/16 \}$ and $\{ -9/4 \}$ are the exceptional values sets in each
case. A complete classification as in \cref{prop:lomu} seems to be currently
out of reach though.
Identities such as 
\[e^{-z} \, {}_2F_2[1/2,1/3;-1/2,-2/3;z] = 1-z/2+3z^2\]
(implying that the third derivative of the left-hand side is identically
zero) show that an assumption corresponding to the assump\-tion $a-b\notin
\mathbb N$ in \cref{prop:lomu} is obviously necessary on the rational
parameters of the ${}_2F_2[a,b;c,d;z]$: to avoid trivial situations, besides
the fact that the parameters $a,b,c,d$ must all be in $\mathbb Q
\setminus\mathbb Z_{\le 0}$, we must not have $a-c\in \mathbb N$ and $b-d\in
\mathbb N$, or $a-d\in \mathbb N$ and $b-c\in \mathbb N$ (note that in the
second example above, we have $3/4-(-9/4)\in \mathbb N$ but $1/4-5/4\notin
\mathbb N$, while $1/4-(-9/4)\notin \mathbb N$ and $3/4-5/4\notin \mathbb
N$).

\subsection{An example with Gauss' hypergeometric function}

The approach leading to special evaluations is very general and not restricted to $E$-functions. For instance, Gauss' hypergeometric function ${}_2F_1[a,b;c;z]$ satisfies a differential equation whose adjoint is solved by $R(z)\coloneq {}_2F_1[1-a,1-b;2-c;z]$. The approach from \cref{ssec:nhm} then deduces that the hypergeometric function satisfies a first order \emph{inhomogeneous} equation, with coefficients that are not polynomials in general, namely
\[z(z-1)R(z)y'(z)+(z(1-z)R'(z)+((a+b-1)z+1-c)R(z))y(z)+c-1=0.\]
It follows that if $\rho$ is a simple zero of~$R(z)$ different from 0,1, one gets the special evaluation
\[{}_2F_1[a,b;c;\rho]=\frac{1-c}{\rho(1-\rho)R'(\rho)}.\]
The special case  $c=a+k+1$ ($k\in\mathbb N$) gives a nice analogue of \cref{prop:newfam1F1}. To state it, recall that the $k$th \emph{Jacobi polynomial} $P_k^{(\alpha,\beta)}$ with parameters $\alpha,\beta\in\mathbb{C}$ is defined by
\[
P_k^{(\alpha, \beta)}(z)\coloneq  2^{-k} \cdot \sum_{j=0}^k
\binom{k+\alpha}{k-j}
\binom{k+\beta}{j}
(z-1)^j (z+1)^{k-j}.
\]
It is classical~\cite[\S6.72]{Szego75} that $P_k^{(\alpha, \beta)}$ has only simple roots (which are even real and in the interval $(-1,1)$ if $\alpha$ and $\beta$ are both real and greater than $-1$), with the notable exception of $\pm 1$ which is a multiple root of $P_k^{(\alpha, \beta)}$ if one of the parameters $\alpha, \beta$ is in $\{-1,\ldots,-k \}$.

\begin{proposition}\label{prop:newfam2F1}
Let $a\in\mathbb Q\setminus\mathbb Z_{\le0}$, $b\in\mathbb Q$ and $k\in\mathbb N$.
If $\rho\in\Qbar\setminus\{ 0,1 \}$ is a root of the polynomial $P_{k}^{(-k-a,b-k-1)}(1-2z)$, then
\begin{equation} \label{eq:special2F1}
{}_2F_1[a,b;a+k+1;\rho]=
\frac{(-1)^ka\binom{a+k}{k}(1 - \rho)^{k - b}}{(k + a - b)P_k^{(-k - a, b - k)}(1-2\rho)}
.
\end{equation}
\end{proposition}
\begin{proof}
With $c=a+k+1$, the hypergeometric function $R(z)$ becomes
\[R(z)={}_2F_1[1-a,1-b;1-a-k;z]=\frac{(-1)^{k} k !}{(a)_{k}}{(1-z)^{b -k -1} P_{k}^{(-k -a , b -k -1)} (1-2 z)}.\]
Its roots different from 0,1 are the roots of $P_{k}^{(-k -a , b -k -1)} (z)$ different from $-1,1$, which are all simple. The formula for the denominator comes from the derivative
\begin{multline*}
R'(z)=-\frac{(-1)^k}{z\binom{a+k-1}k}\left(
 (k+a-b)(1-z)^{b-k-1}P_k^{(-k-a,b-k)}(1-2z)\right.\\
\left. +
 ((a-1)z+b-a-k)(1-z)^{b-k-2}P_k^{(-k-a,b-k-1)}(1-2z)
 \right).\qedhere
\end{multline*}
\end{proof}

\medskip 
\begin{remark} \label{rem:2F1}
	\emph{Let us conclude with a few remarks on \cref{prop:newfam2F1}.}
\end{remark} 

\begin{enumerate}[left=0pt]

\smallskip 
\item 
If none of $a,b,a-b$ is an integer, then $f(z) = {}_2F_1[a,b;a+k+1;z]$ is a transcendental function~\cite{Vidunas07}. Therefore, the evaluation in \cref{eq:special2F1} provides very simple particular cases of algebraic values taken by transcendental $G$-functions at algebraic points. Note that \eqref{eq:special2F1} holds more generally for $a\in\mathbb C\setminus\mathbb Z_{\le0}$ and $b\in\mathbb C$.

\smallskip 
\item As in the case of \cref{prop:newfam1F1}, another proof of \cref{prop:newfam2F1} relies on a relation analogous to \cref{eq:id-1f1} between ${}_2F_1[a,b;a+k+1;z]$,
${}_2F_1[a,b;a;z]=(1-z)^{-b}$ and 
\[
{}_2F_1[a,b;a+1;z]=az^{-a}B_z(a,1-b)=az^{-a}\int_0^z{z^{a-1}(1-z)^b\,dz},
\]
an incomplete beta function. The relation is obtained from those two by repeated use of a contiguity relation.

\smallskip 
\item As in \cref{rem:1F1}, nice special cases of \eqref{eq:special2F1} can be obtained by studying triples 
$(a,b,k)$ for which the Jacobi polynomial $P_k^{(-k - a, b - k-1)}$ factors non-trivially. For instance, the triples 
$(2/5, 3/5, 5)$ and $(2/3,7/3,3)$
yield the evaluations
\begin{align*}
{}_2F_1 \left[ \frac25, \frac35; \frac{32}{5}; \frac12 \right]& = \frac{1683}{2500} \,  \sqrt[5]{8}
\\
\intertext{and}
{}_2F_1\left[ \frac23, \frac73; \frac{14}{3}; \frac{3 \sqrt{5}-5}{2}\right]
&
= \frac{44}{27 \sqrt[3]{28-12 \sqrt{5}} } .
\end{align*}
(Here, by item (1) above, both hypergeometric functions are transcendental.)

\end{enumerate}

\section{Implementation}\label{sec:implem}

\subsection{Minimization is simpler than factorization}\label{sec:ex-factor}
The following simple example illustrates the difference between minimization and factorization.  Take
\[A=z^2\partial_z+3,\quad B=(z-10)\partial_z+z^5\]
and their product
\[C=AB=z^2(z-10)\partial_z^2+(z^7+z^2+3z-30)\partial_z+z^5(5z+3).\]
The computation of a bound on the degree of the coefficients of a
factor of order~1 of~$C$ gives~$10^5+2$. (More generally, changing
$10$ into a large~$N$ leads to a bound~$N^5+2$.) This leads to a large
computation when trying to
factor $C$ without further information. With the extra knowledge that we are looking for a solution of $C$ with initial condition~$y(0)=1$, we easily compute the first~20 coefficients of the unique series solution~$S$ of~$C$ with~$S(0)=1$ and then compute an approximant basis (a Hermite-Padé approximant) of~$(S,S')$ at order~20. This recovers the operator~$B$. It is easily checked that~$B$ is a right divisor of~$C$ by Euclidean division. It follows that all solutions of~$B(y)=0$ such that~$y(0)=1$ are solutions of~$C(y)=0$ and by uniqueness this proves that~$B(y)=0$ is the minimal homogeneous differential equation for~$S$. This takes less than a second with our implementation.
On this example, Maple's factorization routine \textsf{DEtools[DFactor]} has to be killed after running for 1~hour, a further indication that factorization is more complex than minimization.

\subsection{Implementation aspects}

The main difficulty is to avoid the computation of high-order
expansions of power series with rational coefficients. A first gain is achieved by performing most of the computation modulo a sufficiently large prime number (we take a 31-bit long prime number). When a factor is found with modular coefficients, then the actual degree bounds from that factor are used to determine how many rational coefficients of the power series have to be computed and then obtain the differential operator with rational coefficients. 

Another place where time can be saved is in the optimization problems. The computation of an approximant basis returns a linear differential operator of small order if one exists with the given degree bounds. Thus the computation of a tight bound on the number of apparent singularities is only useful if it leads to a bound on the degrees that is smaller than a previously known one. One can therefore add an extra inequality to the optimization problem so that the solver does not waste time in computing an optimum which is larger than what is already known.

In the computation of algebraic values of $E$-functions by 
\cref{algo:alg-values}, the matrix returned by Beukers' 
\cref{algo:Beukers} is not needed in full, as only its value
at~$z=\alpha$ is of interest. Thus, one should instead execute 
\cref{algo:Beukers} over Laurent expansions in powers of~$(z-\alpha)$,
increasing the precision of intermediate computations until the result
is found.

\subsection{Timings}

\begin{table}[t]
    \centering
    \begin{tabular}{cccccccc}
        $(m,p)$ & (ord,deg) & (ord,deg) & (ord,deg)& number& number&
        total&time\\
                & rec&
                original&minimal&modular&rational&time&cert.\\
                &    & diff.eq.&diff.eq. &terms&terms&(s.)&\\
                \hline \\[-4mm]
     (1,5)&(6,32)&(32,29)&(12,29)&2461&408&13.&26\%\\
                \hline \\[-4mm]
     (2,4)&(5,26)&(26,24)&(11,24)&1501&317&7.2 &26\%\\
     (2,5)&(6,44)&(44,41)&(15,41)&5992&693&52.&21\%\\
                \hline \\[-4mm]
     (3,3)&(6,28)&(28,25)&(12,29)&2461&408&13. &25\%\\
     (3,4)&(7,51)&(51,47)&(16,47)&8288&838&92. &22\%\\
     (3,5)&(8,76)&(76,72)&(20,72)&20702&1559&358. &30\%\\
                \hline \\[-4mm]
     (4,2)&(5,27)&(27,24)&(12,24)&1927&343&10. &27\%\\
     (4,3)&(6,41)&(41,38)&(15,38)&5362&645&35. &29\%\\
     (4,4)&(6,46)&(46,43)&(18,47)&9634&936&109.&24\%\\
     (4,5)&(8,92)&(92,88)&(24,88)&37228&2255&983.&21\%\\
                \hline \\[-4mm]
     (5,1)&(6,32)&(32,29)&(12,29)&2461	&408	&14.&22\% \\
     (5,2)&(7,51)&(51,47)&(16,47)&2288&838&92.&24\%\\
     (5,3)&(8,76)&(76,72)&(20,72)&20702&1559&477.&27\%\\
     (5,4)&(9,109)&(109,104)&(25,104)&50064&2761&1534.&23\%\\
     (5,5)&(10,134)&(134,129)&(30,145)&103024&4562&5216.&23\% \\
                \hline \\[-4mm]
    \end{tabular}
    \caption{Experimental results (experiments with smaller values of
    $(m,p)$ that complete under 3 sec. are not
    listed)}
    \label{tab:experiments}
\end{table}

Experimental results\footnote{The timings were obtained with Maple2021 on a 2018 Mac~mini.} on the family of power series
\[
f_{m,p} (z) = \sum_{n=0}^\infty \left( \sum_{k=0}^n \binom{n}{k}^m \binom{n+k}{k}^p \right) \frac{z^n}{n!}
\]
are reported in \cref{tab:experiments}. 
These power series are exponential generating functions of Apéry-like sequences, hence they are $E$-functions by design.
The case $(m,p)=(2,1)$ was considered by Adamczewski and Rivoal in~\cite[p.~706]{AdamczewskiRivoal2018}, who proved that $f_{2,1}$ is a purely transcendental $E$-function.
We used our algorithms to reprove this result and to extend it to other values of $m$
 and $p$, see \cref{tab:experiments}.
 
For each $(m,p)$, we indicate the order and the degree of the 
(minimal) recurrence computed by Zeilberger's algorithm, the order and
the degree of the differential equation deduced from this recurrence,
and those of the minimal-order (homogeneous) differential equation
obtained by our implementation when run on this differential equation.
We also give the number of coefficients of the sequence that were
computed modulo a prime number and the number of rational coefficients
of the sequence that were computed. 
The next column contains the time
in seconds spent in the whole computation (homogeneous and
inhomogeneous minimization and proof that there are no exceptional
values). Finally,
the last column gives the proportion of that time spent in
certifying the minimality of the differential equation that has
been computed. 

In more detail, during the computations, the most time-consuming
operation is that of the approximant basis over rational coefficients,
i.e., the reconstruction of the minimal differential equation itself.
This takes between~$1/3$ and~$1/2$ of the time, depending on the
examples.
Next comes the computation of approximant bases modulo a prime
number: even though the order of the power series is much larger in
the certification phase, modular computation makes it faster. This
takes between~$1/5$ and~$1/3$ of the time.
Finally,
the third most expensive part is the
computation of the rational coefficients needed to reconstruct the
operator. In the last case $(m,p)=(5,5)$, this is even more expensive
than the other two steps. The other operations needed in these tests, such as
computing gcrds, finding a minimal inhomogeneous equation,  computing
Beukers' matrix and its kernel, are all negligible compared to
these three.

The time spent in the certification of the
minimality that is displayed in the last column is consistently
between~$20\%$ and~$30\%$ of the total time. This overlaps with the
proportions given above since this time contains the computation of a
large number of modular coefficients and the computation of an
approximant basis.

\medskip 
On the basis of these experimental results, we ask the following questions on the family $f_{m,p}(z)$ and leave them for further research. The data in \cref{tab:experiments}, plus a few more experiments (not included in \cref{tab:experiments}), are in favor of positive answers to all these questions.

\begin{question}
Is the $E$-function $f_{m,p}$ purely transcendental for any $m\geq 1$ and $p \geq 1$?
\end{question}

\begin{question}
Is it true that for any odd~$m$, the minimal-order linear differential equation $L^{\textrm{min}}_{m,p}(y)=0$ 
satisfied by $f_{m,p}$ has order
$\lfloor (N+1)^2/4 \rfloor$ and degree
$\lfloor N (2N^2-3N + 4) /12 \rfloor$, where $N=m+p$?
In particular, when both $m$ and $p$ are odd, is it true that 
$\ord(L^{\textrm{min}}_{m,p}) = \ord(L^{\textrm{min}}_{p,m})$
and 
$\deg(L^{\textrm{min}}_{m,p}) = \deg(L^{\textrm{min}}_{p,m})$?
\end{question}

\section*{Acknowledgements.}
This work has been supported by the French ANR grant
\textcolor{magenta}{\href{https://specfun.inria.fr/chyzak/DeRerumNatura/}{DeRerumNatura}},
ANR-19-CE40-0018.
We thank Frits Beukers, Lucia Di Vizio, Marc Mezzarobba, Gaël Rémond and Michael Singer for many interesting discussions.
We are also grateful to the anonymous reviewers for their
useful and constructive comments.


\newcommand{\etalchar}[1]{$^{#1}$}
\providecommand{\bysame}{\leavevmode\hbox to3em{\hrulefill}\thinspace}
\providecommand{\MR}{\relax\ifhmode\unskip\space\fi MR }
\providecommand{\MRhref}[2]{%
  \href{http://www.ams.org/mathscinet-getitem?mr=#1}{#2}
}
\providecommand{\href}[2]{#2}

\end{document}